\documentclass[12pt]{article}
\usepackage[final]{epsfig}
\usepackage{amssymb}
\usepackage{latexsym}
\usepackage{wrapfig}
\usepackage{amsmath}
\usepackage{booktabs}
\usepackage{threeparttable}
\usepackage{xcolor}
\usepackage{amsthm}
\usepackage{comment}
\usepackage{graphicx, caption, subcaption}
\usepackage{authblk}
\usepackage{titlesec}
\allowdisplaybreaks

\newcommand{\beq}{\begin{equation}}
\newcommand{\eeq}{\end{equation}}
\newcommand{\beqs}{\begin{eqnarray}}
\newcommand{\eeqs}{\end{eqnarray}}

\newtheorem{theorem}{Theorem}[section]
\newtheorem{corollary}{Corollary}[section]
\newtheorem{lemma}{Lemma}[section]
\newtheorem{proposition}{Proposition}[section]

\DeclareMathOperator{\sign}{sign}

\titleformat*{\section}{\centering \large\bfseries}
\titleformat*{\subsection}{\centering \large\itshape}

\title{The designs and deformations of rigidly and flat-foldable quadrilateral mesh origami}
\author[1]{Fan Feng}
\author[2]{Xiangxin Dang}
\author[3]{Richard D. James}
\author[4]{Paul Plucinsky}
\affil[1]{Cavendish Laboratory, University of Cambridge, UK}
\affil[2]{State Key Laboratory for Turbulence and Complex System, Department of Mechanics and Engineering Science, College of Engineering, Peking University, Beijing 100871, China}
\affil[3]{Aerospace Engineering and Mechanics, University of Minnesota, USA} 
\affil[4]{Aerospace and Mechanical Engineering, University of Southern California, USA}

\begin{document}

\maketitle 

\begin{abstract}
Rigidly and flat-foldable quadrilateral mesh origami is the class of quadrilateral mesh crease patterns with one fundamental property:  the patterns can be folded from flat to fully-folded flat by a continuous one-parameter family of piecewise affine deformations that do not stretch or bend the mesh-panels.   In this work, we explicitly characterize the designs and deformations of all possible rigidly and flat-foldable quadrilateral mesh origami.  Our key idea is a rigidity theorem (Theorem \ref{DesignTheorem}) that characterizes compatible crease patterns surrounding a single panel and enables us to march from panel to panel to compute the pattern and its corresponding deformations explicitly.  The marching procedure is computationally efficient.  So we use it to formulate the inverse problem: to design a crease pattern to achieve a targeted shape along the path of its rigidly and flat-foldable motion. The initial results on the inverse problem are promising and suggest a broadly useful engineering design strategy for shape-morphing with origami.

%
\end{abstract}

\section{Introduction}

Origami is a popular Japanese artform that aims to achieve complex shape by the intricate folding of an initially flat piece of paper.  The state of the art in origami design incorporates mathematical principles for packing problems and well-known fold operations into algorithms capable of producing crease patterns that can be folded into a wide range of targeted shapes \cite{demaine2000folding, demaine2017origamizer, lang1996computational, tachi2009origamizing}.  While striking in their achievements, the crease patterns emerging from these algorithms need to be folded by a delicate sequence of (often numerous) steps to reach the final shape.  As an artform, this is hardly an issue to the skilled Origamist. However, the use of origami in engineering involves modalities of folding such as simple external loading, motors, active materials, hinge offsets for thick origami panels \cite{Chen396}, and hinge mechanisms \cite{schenk2011origami, FILIPOV201726}; none of which  demonstrate capabilities anywhere close to that of an Origamist. As a result, the utility of these algorithms for practical engineering is limited to the simplest of designs and shapes \cite{felton2014method}.

In contrast, the famed Miura-Ori  pattern shown in Fig.\;\ref{fig:MiuraOri} was adapted as a concept for packaging and deploying large space membranes \cite{miura1985method}, and has since been the topic of numerous investigations in the engineering literature.  This pattern is made of a repeating corrugated unit cell (in bold) whose crease ``mountain-valley assignments" are biased in the flat state (red/blue).   Because of the geometry of the unit cell and this bias,  the Miura-Ori has a unique (up to rigid motions), continuous one-parameter family of deformations that do not stretch or bend any of the panels; that of a uniform compressive motion from flat to folded flat.  The Miura-Ori crease pattern is  also inherently confined to a thin sheet.  Thus, geometry and elasticity \cite{conti2008confining,friesecke2002theorem} dictate that a carefully designed Miura-Ori, involving stiff panels and flexible creases, can achieve this large coordinated motion by a variety of stimuli in a fashion robust to non-idealities.  In fact, every modality discussed above has proven effective at folding a Miura-Ori\textemdash most at a variety of length and time-scales\textemdash and this made the Miura-Ori a popular template for exploring functionality enabled by shape-morphing: for tunable and novel properties in  compression \cite{silverberg2014using}, bending \cite{wei2013geometric}, and as a meta-material \cite{schenk2013geometry}; for self-folding at the microscale \cite{na2015programming} and mesoscale \cite{tolley2014self}; for physically driven self-origanization \cite{mahadevan2005self}; for a fast and controlled response in soft-robotics \cite{kim2018printing}; and for ease of deployment in space structure applications \cite{miura1985method, pellegrino2014deployable}.  

\begin{figure}
\centering
\includegraphics[width = 6.5in]{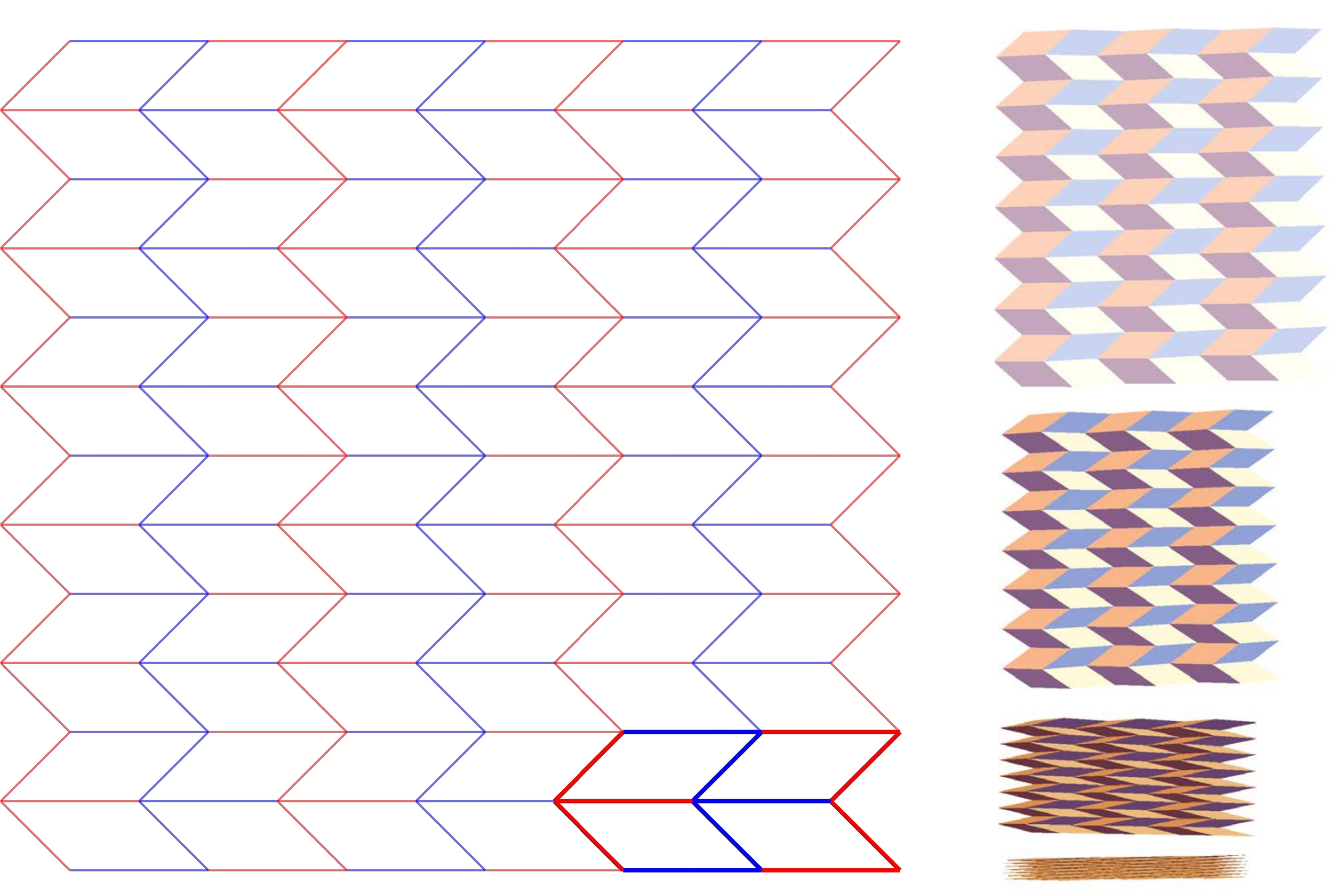}
\caption{The Miura-Ori crease pattern and its rigidly foldable motion.}
\label{fig:MiuraOri}
\end{figure}

All this notwithstanding, only so much functionality can be squeezed out of the shape-morphing capability of a Miura-Ori.  Yet, existing predictive algorithms for producing complex shapes by sequential folding, as discussed above, are quite limited.  We aim to fill the gap.  To this point, we note that the Miura-Ori is but the most basic example in the class of origami termed \textit{rigidly and flat-foldable quadrilateral mesh origami} (RFFQM). This class is composed of quadrilateral mesh crease patterns with the defining property: when the crease pattern is suitably biased in the flat state, by imposing the correct mountain-valley assignments, the pattern can fold from flat to fully-folded flat by a unique (up to rigid motion) and continuous one-parameter family of piecewise-affine deformations that do not stretch or bend the mesh-panels.  As this shape-morphing capability is desirable in  many of the aforementioned applications, we find it natural to examine the configuration space of RFFQM in detail. 

The  understanding of RFFQM has evolved significantly within the past several years.  Kawasaki  \cite{kawasaki1991relation} asserted a necessary condition for flat-foldability at a single degree-4 vertex, proved by Hull \cite{hull1994mathematics}.  Huffman parameterized the kinematics at flat-foldable degree-4 vertices \cite{huffman1976curvature}. Belacastro and Hull \cite{hull2002modelling} formalized conditions of compatibility for origami using piecewise affine transformations.  Tachi \cite{tachi2009generalization} revitalized the subject by isolating the fundamental questions of compatibility and solving them in special cases by perturbative methods, yielding surprising examples.   Finally, Lang and Howell \cite{lang2018rigidly} made clever use of observations due to Tachi \cite{tachi2010geometric} to derive an algorithm that parameterizes the configuration space of RFFQM.

In this paper, we develop some new lines of thinking, while expanding on many ideas introduced previously. On the latter, it turns out that every fundamental question in origami design\footnote{Except for open issues of global invertibility, which we do not discuss.} can be reduced to  the question of ``how to link up neighboring folding angles by parameterizations that depend, in a delicate way, on the geometry of the crease pattern".   To the adept Origamist, like those referenced previously, this is apparently a natural matter of fact.  As such, it serves as a starting point to any development in the literature.  From the point of view of continuum mechanics, however, the more natural observation is: \textit{any} origami\textemdash in the purest sense of the definition\textemdash is the result of an explicit deformation of an initially flat crease pattern.  Thus, questions and answers about designing origami crease patterns and parameterizing their kinematics can be addressed by working directly with deformations; specifically, by deriving and characterizing the local compatibility conditions on deformation gradients that are necessary and sufficient for continuity of the overall deformation.  

The use of the deformation as the basic kinematic object is an underdeveloped aspect of origami design.  We develop it in Section \ref{sec:LocalKinematics} for the following purposes:
\begin{itemize}
\item  to derive and characterize  \textit{vertex compatibility}, i.e., the conditions under which the crease pattern around a single vertex can be folded as origami;
\item  to derive \textit{panel compatibility}, i.e., the conditions under which a crease pattern surrounding a single panel can folded as RFFQM;
\item to explicitly characterize the full kinematics of RFFQM when such compatibility is achieved.
\end{itemize}
The last point is a uniquely beneficial aspect of working directly with deformations. By this characterization, we fully determine  the deformations surrounding a single vertex and surrounding a single panel up to an overall rigid body motion of the entire crease pattern.  The global deformation of the overall crease pattern must then be consistent with this local deformation. 


Formalizing these necessary and sufficient conditions for compatibility brings several observations and results new to the origami literature. We first show that parameterizations of the folding angles  that solve vertex compatibility belong to an underlying Abelian group. We then use properties of this group to characterize panel compatibility in full generality.  The key ideas  are: 1) Panel compatibility is a condition that a certain composition of maps in this group is the identity map, and 2) a fundamental property of the group is that composition of maps is, in a precise sense, the same as multiplication of parameters.  Thus, an apparently hard problem of characterizing a composition of maps is reduced to an easier problem of characterizing a multiplication of parameters.   From this reduction, we derive a design theorem for RFFQM which, stated informally, addresses the following question: If we know the sector angles and mountain-valley assignments at three of four vertices surrounding the panel, can we design the fourth vertex so that the panel's crease pattern can deform as RFFQM?  The answer is generically yes. In fact, there is a unique characterization of the fourth vertex, which we provide via explicit design formulas in Theorem \ref{DesignTheorem}.  

Because this   theorem can be iterated, it enables a broad engineering design principle for RFFQM (Section \ref{sec:Applications}).  
 First, we explicitly parameterize the deformations of all possible RFFQM by marching from panel to panel and repeatedly applying the design theorem.  This procedure results in a computationally efficient marching algorithm for the configuration space formally similar to Lang and Howell's \cite{lang2018rigidly}, but with distinctly different parameterization.  Our parameterization involves purely referential quantities of the flat crease pattern (sector angles and mountain-valley assignments) whereas theirs involves kinematic variables (folding angles).  We find that the referential quantities are natural variables to use in formulating the inverse problem. Specifically, we incorporate this algorithm into an optimization scheme for producing a crease pattern to achieve a targeted shape along the path of its rigidly and flat-foldable motion.  Finally, we illustrate the power of this approach by studying slight perturbations of the  Miura-Ori crease pattern in this context. Note, Dudte et al.\;\cite{dudte2016programming} achieved a variety of curvatures by perturbed Miura-Ori, but their patterns are not RFFQM and can only be deployed by slightly stressing the panels.  In contrast, our strategy works with perturbations consistent with RFFQM only, yielding patterns that are (ideally) stress-free in deployment.  Importantly, even under this deployability consideration, we are still able to achieve striking examples of  origami that approximate  surfaces of positive Gaussian curvature, negative Gaussian curvature and changes in Gaussian curvature. 

With this work, we seek to introduce a continuum mechanics approach to origami and bring general and easily implementable methodologies for engineering design to the forefront. Our overarching motivation is that applications ranging from  space structure design \cite{arya2016packaging,arya2016ultralight},  to biomedical devices \cite{kuribayashi2006self} and soft robotics \cite{kim2018printing} often seek a targeted final shape and the ability to achieve it by  \textit{controlled deployment} from an easily manufactured or highly compact state.   In RFFQM, we have an easily manufactured state (the flat state), a highly compact state (the folded flat state), and controlled deployment by a single degree of freedom motion; in other words, we have a template for shape-morphing in engineering design.   While we are in the early stages of examining the versatility and utility of these methods, the results here\textemdash and those in forthcoming work on the inverse problem \cite{dang_inverse}\textemdash point towards a broad engineering design tool for shape-morphing with RFFQM.

\textbf{Notation.} Throughout, we deal with real valued scalars $\gamma \in \mathbb{R}$, vectors $\mathbf{v} \in \mathbb{R}^3$ and matrices  $\mathbf{F} \in \mathbb{R}^{3\times3}$.  We let $\mathbb{S}^2$ denote the set of all unit vectors on $\mathbb{R}^3$ and $SO(3)$ denote the set of all rotation matrices in $\mathbb{R}^{3\times3}$, called simply rotations below.  The vectors $\{ \mathbf{e}_1, \mathbf{e}_2, \mathbf{e}_3\}$ will always refer to the standard orthonormal basis on $\mathbb{R}^3$.   We study exclusively deformations which do not bend or stretch the mesh-panels within a crease pattern, i.e., we study \textit{origami deformations}.  Denote an overall crease pattern by $\Omega \subset \mathbb{R}^2$ and let $\mathbf{y} \colon \Omega \rightarrow \mathbb{R}^3$ be an origami deformation of this crease pattern.  Then, $\mathbf{y}$  is defined by the following two properties:
\begin{itemize}
\item If we isolate any single mesh-panel $\mathcal{P} \subset \Omega$, then the deformation on this panel has the form 
\begin{equation}
\begin{aligned}\label{eq:origamiDefinition}
\mathbf{y}(x_1, x_2) = \mathbf{R}  \mathbf{x} + \mathbf{b}, \quad \mathbf{x} = x_1 \mathbf{e}_1 + x_2 \mathbf{e}_2 \in \mathcal{P}
\end{aligned}
\end{equation}
for some rotation $\mathbf{R} \in SO(3)$ and translation $\mathbf{b} \in \mathbb{R}^3$.
\item  The rotation and translation can change from panel to panel but only in such a way that the deformation is continuous on the entire crease pattern $\Omega$.
\end{itemize}
We often refer to the rotations in the formula (\ref{eq:origamiDefinition}) as \textit{deformation gradients}.  This is a slight abuse of notation.  Since $\mathbf{y}$ maps a subset of $\mathbb{R}^2$ to $\mathbb{R}^3$, the deformation gradients of origami are actually linear transformations of the form $\widetilde{\mathbf{R}} \in O(2,3) := \{ \widetilde{\mathbf{F}} \in \mathbb{R}^{3\times2} \colon \widetilde{\mathbf{F}}^T \widetilde{\mathbf{F}} = \mathbf{I} \}$.   For our purposes, however, this distinction is unimportant, since
a $3 \times 2$ component of the rotation matrix determines uniquely its third column as the cross product of the first two.  So we prefer to work consistently here with vectors on $\mathbb{R}^3$ and matrices on $\mathbb{R}^{3 \times3}$.  We refer the interested reader to Conti and Maggi \cite{conti2008confining} for a  precise description of the definitions and function spaces for origami deformations.

\section{Local kinematics}\label{sec:LocalKinematics}

\subsection{Compatibility at a single vertex.}

We begin with a local description of the kinematics of degree-4 vertices.  To motivate this study, note that, by isolating a small ball around a single degree-4 vertex within an overall quad-mesh crease pattern (Fig.\;\ref{fig:SingleVertex}), we can deduce local properties on the kinematics of RFFQM.  It is the consistency of these local conditions that allows us to systematically investigate the global kinematics of the crease pattern.  

\begin{figure}
\centering
\includegraphics[width = 6.5in]{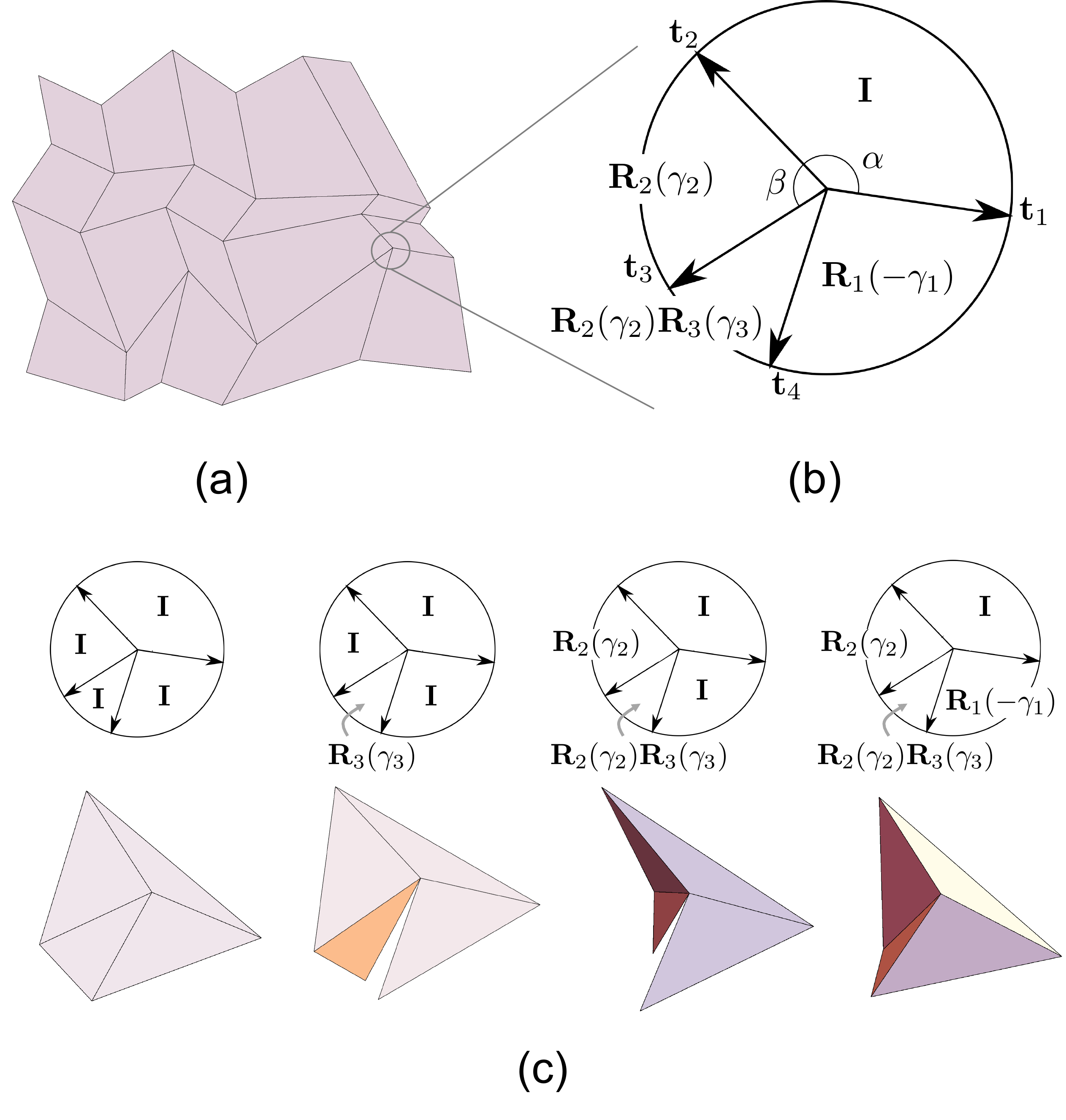}
\caption{A description of vertex compatibility. (a) Isolate a single vertex within an overall crease pattern.  (b) Its kinematics is, up to an overall rigid motion, necessarily described by rotations about crease tangent as depicted, and these rotations are subject to a vertex compatibility condition. (c)  This  can be understood by the following thought experiment: 1) cut the pattern at a crease, 2) apply sequences of rotations along creases to this cut pattern as shown, and 3) parameterize the corresponding rotation angles so that the pattern can be glued back together again.}
\label{fig:SingleVertex}
\end{figure}

As is well-known in the origami literature \cite{hull1994mathematics,lang2011origami}, the necessary and sufficient condition for RFFQM of an initially flat crease pattern composed of a isolated degree-4 vertex (Fig.\;\ref{fig:SingleVertex}(b)) is \textit{Kawasaki's condition} \cite{kawasaki1991relation}, i.e., the condition that opposite sector angles sum to $\pi$.  Since we are exclusively interested in RFFQM in this work, we therefore make repeated use of the following local description of the crease pattern: We label the four crease tangents of this vertex as $\mathbf{t}_i \in \mathbb{S}^2$ with $\mathbf{t}_i \cdot \mathbf{e}_3 = 0$ in a counter-clockwise fashion  as depicted\footnote{We will also use a clock-wise labeling when convenient.}, we define the sector angles as
\begin{equation}
\begin{aligned}
\alpha := \arccos (\mathbf{t}_1 \cdot \mathbf{t}_2), \quad \beta := \arccos (\mathbf{t}_2 \cdot \mathbf{t}_3),
\end{aligned}
\end{equation}
and we enforce Kawasaki's condition
\begin{equation}
\begin{aligned}\label{eq:Kawasakis}
\alpha , \beta \in (0,\pi), \quad \arccos (\mathbf{t}_3 \cdot \mathbf{t}_4) = \pi - \alpha, \quad  \arccos (\mathbf{t}_4 \cdot \mathbf{t}_1) = \pi - \beta.
\end{aligned}
\end{equation}
Note that this is simply a necessary condition on the sector angles for the overall crease pattern to admit RFFQM.

The rigid-folding kinematics of an isolated degree-4 vertex under Kawasaki's condition are also well-known in the literature \cite{lang2011origami, tachi2009generalization}.  The earliest derivation, due to Huffman \cite{huffman1976curvature}, employs spherical trigonometry.  Following the \textit{continuum mechanics} approach 
adopted here, we write explicitly the underlying deformation (from flat).   That is, we impose that the deformation gradient satisfies a jump \textit{compatibility condition} across each crease.  Though less familiar to the origami community\footnote{See, however, \cite{hull2002modelling}.}, this approach has found longstanding application in the continuum mechanics modeling of microstructure for materials undergoing the martensitic phase transformation \cite{ball1989fine, bhattacharya2003microstructure,cui2006combinatorial, song2013enhanced}.   It is also a highly generalizable; the basic structure of the derivation is unchanged if, say, we relax Kawasaki's condition, introduce additional creases, and/or consider non-Euclidean vertices\footnote{These are vertices whose sector angles do not sum to $2\pi$.} \cite{berry2019topological,plucinsky2018patterning, plucinsky2016programming, waitukaitis2019non}. For these reasons, we provide a brief exposition on the topic.

The continuum mechanics approach to the kinematics can be understood heuristically as follows: Rigidly foldable origami describes deformations of the crease pattern that do not  involve any bending or stretching of the panels, i.e, the kinematics are localized at the creases via perfect hinge mechanisms. We can therefore keep one of the panels fixed in its initially flat state by applying $\mathbf{I}$ as the deformation gradient of this region\footnote{As we have done, for instance, with the $\alpha$-sector in Fig.\;\ref{fig:SingleVertex}(b).}.  This has the effect of 1)\;eliminating the degeneracy of rigid body motion and 2)\;providing a convenient reference on which to build the kinematics through deformation gradients that are products of rotations about crease tangents.  Specifically, we will be dealing with rotations $\mathbf{R}_i (\gamma_i) \in SO(3)$ that satisfy
\begin{equation}
\begin{aligned}\label{eq:localRots}
\mathbf{R}_i(\gamma_i) := \mathbf{t}_i \otimes \mathbf{t}_i + \cos \gamma_i \mathbf{P}_i + \sin \gamma_i \mathbf{W}_i \quad \text{ for } \quad \begin{cases}
\mathbf{P}_i := \mathbf{I} - \mathbf{t}_i \otimes \mathbf{t}_i, \\
\mathbf{W}_i := \mathbf{e}_3 \otimes \mathbf{t}_i^{\perp} -  \mathbf{t}_i^{\perp} \otimes \mathbf{e}_3,
\end{cases}
\end{aligned}
\end{equation}
and for $\mathbf{t}_i^{\perp} \in \mathbb{S}^2$ defined by $\mathbf{t}_i^{\perp} := -(\mathbf{t}_i \cdot \mathbf{e}_2) \mathbf{e}_1 + (\mathbf{t}_i \cdot \mathbf{e}_1) \mathbf{e}_1$.  These definitions yield four right-hand orthonormal frames $\{ \mathbf{t}_i, \mathbf{t}_i^\perp, \mathbf{e}_3\}$ and four rotations that satisfy $\mathbf{R}_i(\gamma_i) \mathbf{t}_i = \mathbf{t}_i$.  Accordingly, each rotation describes a right-hand rotation about an axis $\mathbf{t}_i$ by an angle $\gamma_i \in [-\pi,\pi]$. 

 To build the kinematics from these rotations, we follow the sketch in Fig.\;\ref{fig:SingleVertex}(c). From the flat state (far-left), we apply a rotation $\mathbf{R}_3(\gamma_3)$ to the $(\pi - \alpha)$-sector to obtain the configuration (middle-left) which cuts the pattern along the tangent $\mathbf{t}_4$ but otherwise keeps the panels rigid and continuous along every other crease.  We then apply a rotation $\mathbf{R}_2(\gamma_2)$ to the $\beta$-sector and deformed $(\pi - \alpha)$-sector.   This yields a configuration (middle-right) which again keeps the panels rigid and continuous across each of the $\mathbf{t}_{1,2,3}$ creases but  still leaves an opening at the crease with tangent $\mathbf{t}_4$.   We therefore apply a final rotation $\mathbf{R}_1(-\gamma_1)$ to the $(\pi - \beta)$-sector in an attempt to glue the overall rigidly-folded crease pattern back together again.  This question of gluing naturally leads to  jump compatibility conditions on deformation gradients.   Indeed, as the deformation gradients to the left and right of $\mathbf{t}_4$ are $\mathbf{R}_2(\gamma_2) \mathbf{R}_3(\gamma_2)$ and $\mathbf{R}_1(-\gamma_1)$, respectively, the necessary and sufficient conditions for gluing are that the crease tangent $\mathbf{t}_4$ deforms in exactly the same way under the actions of these deformation gradients, i.e., that $\mathbf{R}_2(\gamma_2) \mathbf{R}_3(\gamma_3)\mathbf{t}_4 = \mathbf{R}_1(-\gamma_1) \mathbf{t}_4$.  After some algebra, this equation can be put in the more revealing (completely equivalent) form; that of a vertex compatibility condition
 \begin{equation}
 \begin{aligned}\label{eq:loop}
 \gamma_1, \gamma_2, \gamma_3, \gamma_4 \in [-\pi, \pi]  \quad \text{such that} \quad  \mathbf{R}_1(\gamma_1) \mathbf{R}_2(\gamma_2) \mathbf{R}_3(\gamma_3) \mathbf{R}_4 (\gamma_4)  =\mathbf{I},
 \end{aligned}
 \end{equation}
that concretely links the four folding angles and four compatible deformation gradients at this vertex.

The solutions $ \mathbf{R}_1(\gamma_1), \mathbf{R}_2(\gamma_2),  \mathbf{R}_3(\gamma_3),  \mathbf{R}_4(\gamma_4)$  to vertex compatibility (\ref{eq:loop}) furnish origami deformations of the four panels that merge to make the degree-4 vertex.  These deformations have the form
\begin{equation}
\begin{aligned}\label{eq:localDef}
\mathbf{y}(x_1, x_2) = \begin{cases}
\mathbf{x}  & \text{ if } \mathbf{x} \in \mathcal{P}_{\alpha},  \\
\mathbf{R}_2(\gamma_2) \mathbf{x} & \text{ if } \mathbf{x} \in \mathcal{P}_{\beta}, \\ 
\mathbf{R}_2(\gamma_2) \mathbf{R}_3(\gamma_3) \mathbf{x} & \text{ if } \mathbf{x} \in \mathcal{P}_{\pi - \alpha}, \\
\mathbf{R}_1(-\gamma_1) \mathbf{x} & \text{ if } \mathbf{x} \in \mathcal{P}_{\pi - \beta},
\end{cases}
\end{aligned}
\end{equation}
up to an overall rigid rotation and translation.   Here, $\mathbf{x} := x_1 \mathbf{e}_1 + x_2 \mathbf{e}_2$, and  $\mathcal{P}_{\alpha}$ denotes the panel with sector angle $\alpha$ at the vertex (likewise $\mathcal{P}_{\beta}$, \ldots\; respectively $\beta$, \ldots\; etc.), as sketched in Fig.\;\ref{fig:SingleVertex}(b).  Furthermore, deformations of the form (\ref{eq:localDef}) are continuous if and only if the folding angles satisfy (\ref{eq:loop}). To this point, \textit{every vertex} of a RFFQM  obeys the compatibility condition (\ref{eq:localDef}) for its corresponding parameters\footnote{Note, the four folding angles  $\gamma_1,\gamma_2, \gamma_3, \gamma_4$ and two sector $\alpha, \beta$ angles can change from vertex to vertex.}.  Thus, a complete characterization of the folding angles that satisfy (\ref{eq:loop})  is needed for the designs and deformations of RFFQM.

Before proceeding to the characterization, let us comment on \textit{mountain-valley assignments}.   Physically,  when a deformed crease pattern is viewed from an orientation consistent with the unfolded state, a fold is a valley-fold if ``it looks like a valley" and a mountain-fold if ``it looks like a mountain".  This has a straightforward mathematical interpretation in our formalism\footnote{A counter-clockwise labeling of crease tangents for which the rotations and deformations are defined by these tangents as above.}.  The mountain-valley assignments simply correspond to the domains of the folding angles as  
\begin{equation}
\begin{aligned}\label{eq:MVAngles}
(\text{valley-fold:}) \quad \gamma_i \in (0,\pi), \qquad (\text{mountain-fold:}) \quad \gamma_i \in (-\pi, 0).
\end{aligned}
\end{equation}
Some crease pattern have degenerate folding kinematics where the pattern can be folded-in-half and folded-in-half again\footnote{For instance, a simple checkerboard crease pattern can be folded from flat to folded flat in numerous ways.}.   For RFFQM, we exclude these cases by strictly focusing on the kinematics for which every folding angle is participating in the  deformation by being either a mountain or valley fold.  This is equivalent to enforcing the non-degeneracy condition
\begin{equation}
\begin{aligned}\label{eq:nonDegeneracy}
\sin \gamma_1 \sin \gamma_2 \sin \gamma_3 \sin \gamma_4 \neq 0.
\end{aligned}
\end{equation}


\begin{figure}
\centering
\includegraphics[width = 5in]{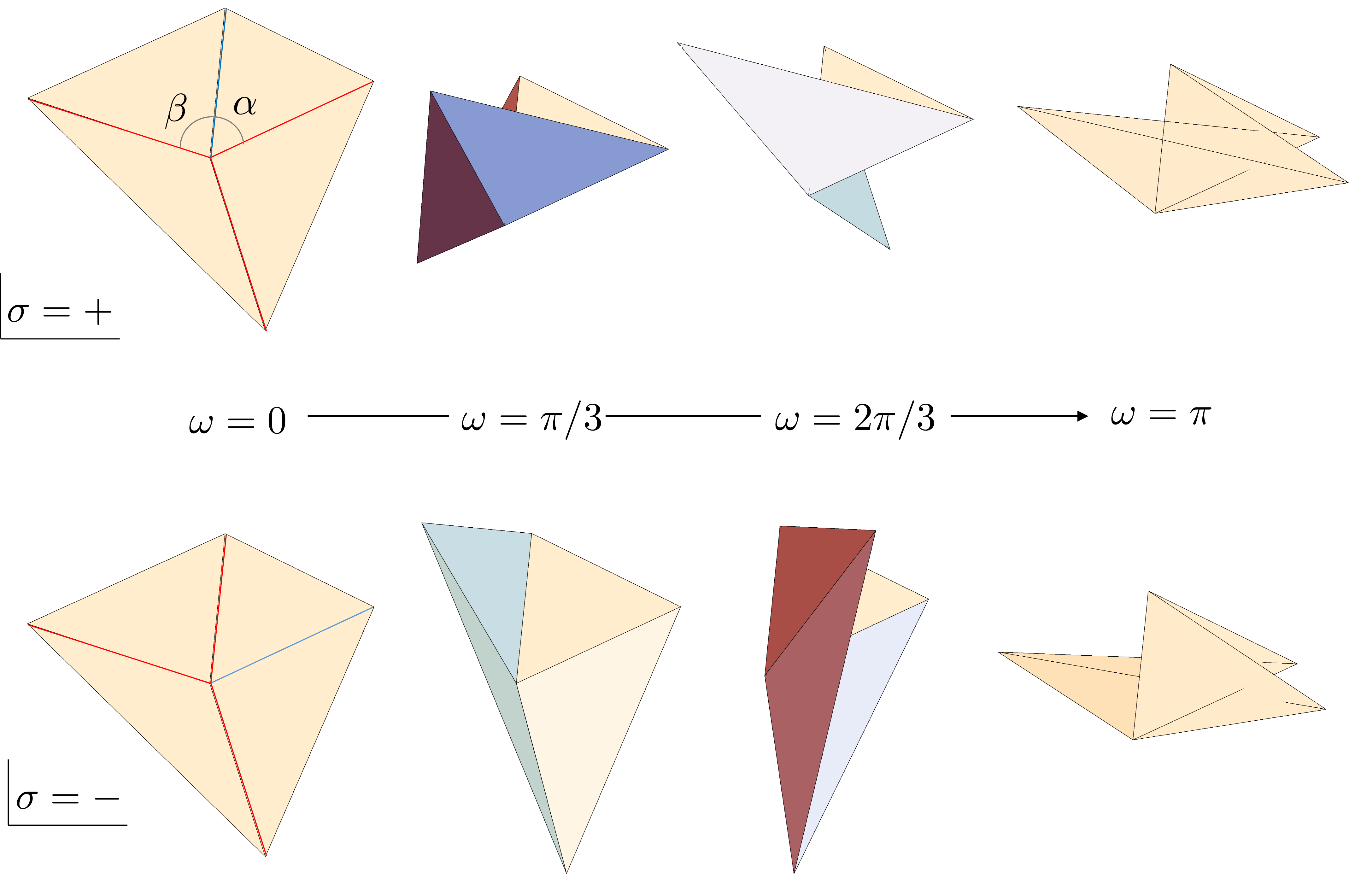}
\caption{The rigid folding kinematic at a degree-4 vertex satisfying Kawasaki's condition.  The origami can fold rigidly along (generically) two mountain valley assignments (red/blue as indicated). It can also be folded from flat ($\omega = 0$) to completely folded flat ($\omega = \pi$) monotonically along these assignments.}
\label{fig:KinSingVert}
\end{figure}

In this setting,  we obtain a precise characterization theorem of the rigid folding kinematics at a single vertex.
\begin{theorem}\label{VertexTheorem}
Assume the sector angles satisfy Kawasaki's condition in (\ref{eq:Kawasakis}) and that $(\alpha, \beta) \neq (\pi/2, \pi/2)$.  Then the folding angles $\gamma_1, \gamma_2, \gamma_3, \gamma_4$ solve vertex compatibility in (\ref{eq:loop}) subject to the non-degeneracy condition (\ref{eq:nonDegeneracy}) if and only if
\begin{equation}
\begin{aligned}\label{eq:foldParams}
&\gamma_1 = \omega, \quad \gamma_2 = \bar{\gamma}_2^{\sigma}(\omega; \alpha, \beta), \quad \gamma_3 = -\sigma \omega, \quad \gamma_4 =  \sigma \bar{\gamma}_2^{\sigma}(\omega; \alpha ,\beta) \\
&\qquad  \text{ for } \quad |\omega| \in (0, \pi) \quad \text{ and } \quad \sigma \in \mathcal{MV}(\alpha, \beta) := \left\{ \begin{array}{l} -  \qquad \text{ if } \alpha = \beta \neq \pi/2 \\ +   \qquad \text{ if } \alpha = \pi - \beta \neq \pi/2  \\ \pm  \qquad  \text{ if }  \alpha \neq \beta \neq \pi - \beta \end{array} \right\} .
\end{aligned}
\end{equation}
In addition, the folding angle function $\bar{\gamma}_2^{\sigma} \colon [-\pi, \pi] \rightarrow [-\pi, \pi]$, $\sigma \in \mathcal{MV}(\alpha, \beta)$, satisfies the explicit relationship
\begin{equation}
\begin{aligned}
\bar{\gamma}_2^{\sigma}(\omega; \alpha , \beta) := \sign \Big((\sigma \cos \beta - \cos \alpha) \omega \Big) \arccos \Big( \frac{(-\sigma1 + \cos \alpha \cos \beta)\cos \omega + \sin \alpha \sin \beta }{-\sigma 1 + \cos \alpha \cos \beta + \sin \alpha \sin \beta \cos \omega } \Big),
\end{aligned}
\end{equation}
and its inverse $\bar{\gamma}_1^{-\sigma} \colon [-\pi, \pi] \rightarrow [-\pi, \pi]$ satisfies the explicit relationship
\begin{equation}
\begin{aligned}\label{eq:inverse}
\bar{\gamma}_1^{-\sigma}(\omega; \alpha, \beta) := \bar{\gamma}_2^{-\sigma}(\omega; \alpha, \pi- \beta),
\end{aligned}
\end{equation}
i.e., $\bar{\gamma}_1^{-\sigma}(\bar{\gamma}_2^{\sigma}(\omega; \alpha , \beta); \alpha, \beta) = \omega$.
Alternatively, if the sector angles satisfy Kawasaki's condition and $(\alpha, \beta) = (\pi/2, \pi/2)$, then there is no solution to (\ref{eq:loop}) under the constraint (\ref{eq:nonDegeneracy}).
\end{theorem}
\noindent We briefly comment on a few properties related to this characterization. By fixing a $\sigma \in \mathcal{MV}(\alpha, \beta)$, the parameterization of the folding angles in (\ref{eq:foldParams}) leads to a continuous one parameter family of origami deformations (\ref{eq:localDef}) with $\omega$ as the folding parameter.  As indicated by the generic example in Fig.\;\ref{fig:KinSingVert}, this parameterization has several universal properties: 1)\;$\omega = 0$ corresponds to the flat state, 2)\;$\omega = \pi$ corresponds to the folded flat state, and 3)\;$\omega \in (0,\pi)$ evolves the folding angles continuously and monotonically from flat to folded flat along a fixed mountain-valley assignment.  Additionally, $\omega \mapsto -\omega$ simply reverses the signs of all the folding angles.  Finally, the solutions $\sigma \in \mathcal{MV}(\alpha, \beta)$ encapsulate all the valid mountain-valley assignments at a single degree-4 vertex.  These always correspond to Maekawa's well known restriction \cite{kasahara1998origami} on the crease assignment that the number of mountain folds differs by two from the number of valley folds, i.e., 
\begin{equation}
\begin{aligned}
\sin \gamma_1 \sin \gamma_2 \sin \gamma_3 \sin \gamma_4 < 0.  
\end{aligned}
\end{equation}
More specifically, a crease  (as highlighted in Fig.\;\ref{fig:KinSingVert}(far-left)) can be an assignment (blue) opposite to the other three creases (red) if and only if its adjacent sector angles sum to a value $< \pi$.  There are generically two such assignments, each indicated by a sign $\sigma =\pm$ associated to the folding angles in the parameterization.  Note though, there are also special cases of vertices ($\alpha = \beta \neq \pi/2$ or $\alpha = \pi - \beta \neq \pi/2$) that have only one valid assignment and completely degenerate vertices ($\alpha = \beta = \pi/2$) with zero valid assignments.  This is what gives rise to the sector angle dependence of the set $\mathcal{MV}(\alpha, \beta)$ in (\ref{eq:foldParams}). 

An early proof of this result is due to Huffman \cite{huffman1976curvature} using the methods of spherical trigonometry.  Espousing the continuum mechanics approach, we provide an alternative proof by working directly and algebraically with the rotation matrices in (\ref{eq:loop}).  
\begin{proof}
Since $\mathbf{R}_{i}(\gamma_i) \mathbf{t}_i = \mathbf{t}_i$ for each $i = 1,\ldots, 4$, a necessary condition for (\ref{eq:loop}) is that 
\begin{equation}
\begin{aligned}
N_1(\gamma_1, \gamma_3) := \mathbf{t}_2 \cdot\big( \mathbf{R}_3(\gamma_3) - \mathbf{R}_1(-\gamma_1)\big) \mathbf{t}_4 =  (\cos \gamma_3 - \cos \gamma_1) \sin \alpha \sin \beta 
\end{aligned}
\end{equation}
should vanish.  Under the non-degeneracy condition (\ref{eq:nonDegeneracy}), $N_1(\gamma_1, \gamma_3) = 0$ if and only if $(\gamma_1, \gamma_3) = (\omega, -\sigma \omega)$ for $|\omega| \in (0,\pi)$ and $\sigma \in \{ \pm \}$.  We assume this parameterization in the remainder as it is necessary for the characterization in the theorem.  Let $\mathbf{P}_2 := \mathbf{I} - \mathbf{t}_2 \otimes \mathbf{t}_2$.  Another necessary condition for (\ref{eq:loop}) under (\ref{eq:nonDegeneracy}) is that 
\begin{equation}
\begin{aligned}
N_2(\gamma_2, \omega, \sigma) := \mathbf{P}_2\Big( \mathbf{R}_2(\gamma_2) \mathbf{R}_3(-\sigma \omega) - \mathbf{R}_1(-\omega)\Big) \mathbf{t}_4 = \mathbf{R}_2(\gamma_2) \mathbf{P}_2  \mathbf{R}_3(-\sigma \omega) \mathbf{t}_4 - \mathbf{P}_2  \mathbf{R}_1(-\omega) \mathbf{t}_4
\end{aligned}
\end{equation}
should vanish.  In this formula, notice that $|\mathbf{P}_2   \mathbf{R}_3(-\sigma \omega) \mathbf{t}_4| = |\mathbf{P}_2 \mathbf{R}_1(-\omega) \mathbf{t}_4|$ since $\mathbf{R}_3(-\sigma \omega) \mathbf{t}_4$ and $ \mathbf{R}_1(-\omega) \mathbf{t}_4$ are unit vectors and since $N_1(\omega, -\sigma \omega) = 0$.   Notice also that $\mathbf{P}_2 \mathbf{R}_1(-\omega) \mathbf{t}_4 \neq 0$ since $|\omega| \in (0,\pi)$.  Thus,  $\mathbf{P}_2  \mathbf{R}_1(-\omega) \mathbf{t}_4$ and $\mathbf{P}_2  \mathbf{R}_3(-\sigma \omega) \mathbf{t}_4$ are two vectors of equal and non-zero magnitude in the plane with normal $\mathbf{t}_2$.   By standard trigonometric identities, we therefore obtain that  $N_2(\gamma_2, \omega, \sigma) = 0$ if and only if the angle $\gamma_2$ satisifies the parameterization 
 \begin{equation}
 \begin{aligned}\label{eq:cosineSineParam}
&\cos \gamma_2 =  \frac{\mathbf{P}_2  \mathbf{R}_1(-\omega) \mathbf{t}_4 \cdot \mathbf{P}_2  \mathbf{R}_3(-\sigma \omega) \mathbf{t}_4}{|\mathbf{P}_2  \mathbf{R}_1(-\omega) \mathbf{t}_4|^2}  = \frac{(-\sigma1 + \cos \alpha \cos \beta)\cos \omega + \sin \alpha \sin \beta }{-\sigma 1 + \cos \alpha \cos \beta + \sin \alpha \sin \beta \cos \omega },  \\
&\sin \gamma_2 = \frac{\mathbf{t}_2 \cdot ( \mathbf{P}_2 \mathbf{R}_3(-\sigma \omega) \mathbf{t}_4 \times  \mathbf{P}_2 \mathbf{R}_1(-\omega) \mathbf{t}_4)}{|\mathbf{P}_2  \mathbf{R}_1(-\omega) \mathbf{t}_4|^2} = \frac{(\sigma \cos \alpha -   \cos \beta) \sin \omega}{-\sigma 1 + \cos \alpha \cos \beta + \sin \alpha \sin \beta \cos \omega }.
 \end{aligned}
 \end{equation}
Here, the latter equalities are by explicit calculation.   For non-degeneracy (\ref{eq:nonDegeneracy}), $\sin \gamma_2$ cannot be zero.  From the parameterization above, this is evidently equivalent to $(\sigma \cos \alpha -   \cos \beta) \neq 0$ since $\sin \omega \neq 0$.   In other words, we conclude that $(\alpha, \beta) \neq (\pi/2, \pi/2)$ and $\sigma \in \mathcal{MV}(\alpha, \beta)$ are necessary for the characterization, as stated in the theorem.    Under these restrictions, we obtain $\gamma_2 = \bar{\gamma}_2^{\sigma}(\omega; \alpha, \beta)$ directly from (\ref{eq:cosineSineParam}).
 
 Now, since $N_1(\omega, -\sigma \omega) = 0$ and $N_2(\bar{\gamma}_2^{\sigma}(\omega; \alpha, \beta), \omega, \sigma) = 0$, we conclude that 
 \begin{equation}
 \begin{aligned}
 \mathbf{R}_1(\omega) \mathbf{R}_2(\bar{\gamma}_2^{\sigma}(\omega; \alpha, \beta)) \mathbf{R}_3(-\sigma \omega) \mathbf{t}_4 = \mathbf{t}_4.  
 \end{aligned}
 \end{equation}
As $\mathbf{R}_1(\omega) \mathbf{R}_2(\bar{\gamma}_2^{\sigma}(\omega; \alpha, \beta)) \mathbf{R}_3(-\sigma \omega)$ is a rotation that fixes $\mathbf{t}_4$, there is a unique $\gamma_4 \in (-\pi, \pi]$ such that $\mathbf{R}_1(\omega) \mathbf{R}_2(\bar{\gamma}_2^{\sigma}(\omega; \alpha, \beta)) \mathbf{R}_3(-\sigma \omega) \mathbf{R}_4(\gamma_4) = \mathbf{I}$.   To obtained the precise parameterization of $\gamma_4$, observe the following: We have parameterized $\gamma_2$ and $\gamma_3$ in terms of $\gamma_1$ by a systematic procedure. We can alternatively repeat the steps of this procedure to parameterize $\gamma_4$ and $\gamma_1$ in terms of $\gamma_3$.  In doing this, we obtain 
 \begin{equation}
 \begin{aligned}
 \gamma_3 = \tilde{\omega}, \quad \gamma_1 = - \tau \tilde{\omega}, \quad \gamma_4 = \bar{\gamma}_2^{\tau} (\tilde{\omega}; \pi - \alpha, \pi - \beta), \quad |\tilde{\omega}| \in (0,\pi), \quad \tau \in \mathcal{MV}(\pi - \alpha, \pi - \beta).
 \end{aligned}
 \end{equation}
 By matching these formula to our original parameterization, we deduce that $\tilde{\omega} = - \sigma \omega$, $\tau = \sigma$ and $\gamma_4 = \bar{\gamma}_2^{\sigma} (-\sigma \omega; \pi - \alpha, \pi - \beta) = \sigma \bar{\gamma}_2^{\sigma}(\omega; \alpha, \beta)$, as in the characterization of the theorem.  Finally, the inverse $\bar{\gamma}_1^{-\sigma}$  of $\bar{\gamma}_2^{\sigma}$ can be deduced by another reparameterization, or it can simply be verified explicitly. 
\end{proof}

\subsection{Compatibility at a single panel.}

Building on the characterization at a single vertex in Theorem \ref{VertexTheorem}, we derive the necessary and sufficient conditions for RFFQM about a single panel (Fig.\;\ref{fig:SinglePanel}).   This characterization of panel compatibility is well-known \cite{tachi2009generalization}, and it corresponds to a consistency condition on the folding angle functions when ``taking a loop" about the panel in question.   We first introduce the result in the form of a theorem.  Then, guided by the sketch in Fig.\;\ref{fig:SinglePanel}(b-d), we briefly provide a derivation using the continuum mechanics approach we have espoused thus far.

\begin{figure}
\centering
\includegraphics[width = 6.in]{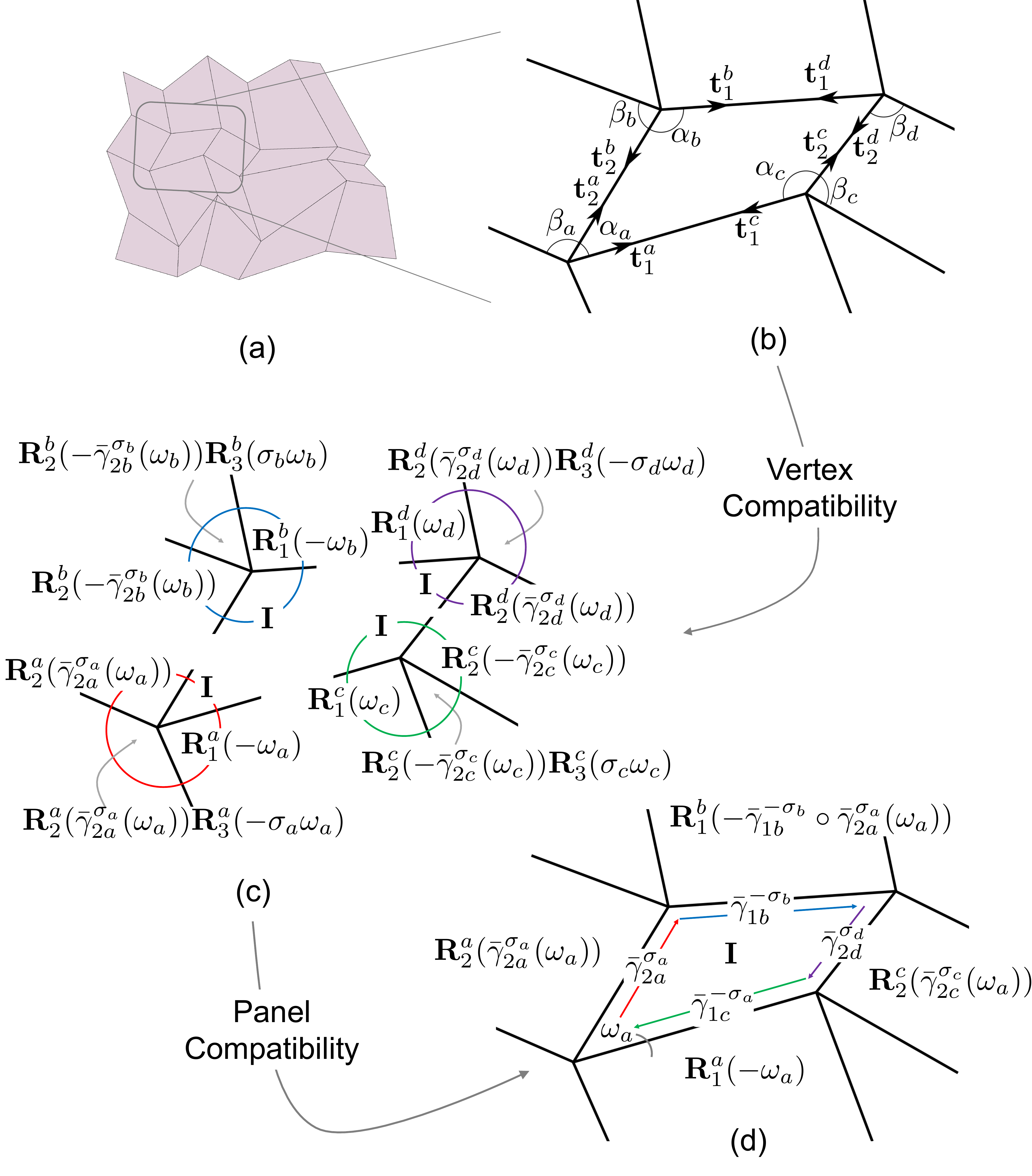}
\caption{Step by step schematic to derive compatibility at a single panel.  (a) Isolate a panel within an overall crease pattern.  (b) Introduce the depicted local notation for the sector angles and crease tangents at the panel.  (c) Directly apply the vertex characterization theorem (i.e., Theorem \ref{VertexTheorem}) to each vertex so as to be consistent with this notation.  (d) The panel compatibility condition is obtained by matching the deformation gradients on vertical and horizontal adjacent panels in (c).}
\label{fig:SinglePanel}
\end{figure}

To introduce the result, we first isolate a crease pattern surrounding a panel and enforce Kawasaki's condition at all the vertices.  Such crease patterns are, in the most general setting, characterized by exactly seven independent sector angles (Fig\;\ref{fig:SinglePanel}(b)).  These angles are restricted to the domains
\begin{equation}
\begin{aligned}\label{eq:sectorPanel}
&\alpha_a, \alpha_b, \alpha_c, \beta_a , \beta_b , \beta_c, \beta_d \in (0,\pi), \quad  2\pi - \alpha_a - \alpha_b - \alpha_c  \in (0, \pi), \\
&(\alpha_a, \beta_a),  (\alpha_b, \beta_b) ,  (\alpha_c, \beta_c), (2\pi -\alpha_a - \alpha_b - \alpha_c, \beta_d)  \neq (\pi/2, \pi/2). 
\end{aligned}
\end{equation}
The desired result then arises as a condition on the folding angle functions that are relevant to the vertices of this panel.  These functions have the form
\begin{equation}
\begin{aligned}
&\bar{\gamma}_{2a}^{\sigma_a}(\omega) := \bar{\gamma}_2^{\sigma_a}( \omega; \alpha_a, \beta_a), \quad \bar{\gamma}_{1a}^{-\sigma_a}(\omega) := \bar{\gamma}_1^{-\sigma_a}( \omega; \alpha_a, \beta_a), \quad \sigma_a \in \mathcal{MV}(\alpha_a, \beta_a), \\
&\bar{\gamma}_{2b}^{\sigma_b} (\omega) := \bar{\gamma}_2^{\sigma_b}( \omega; \alpha_b, \beta_b), \quad \bar{\gamma}_{1b}^{-\sigma_b}(\omega) := \bar{\gamma}_1^{-\sigma_b}( \omega; \alpha_b, \beta_b), \quad \sigma_b \in \mathcal{MV}(\alpha_b, \beta_b), \\
&\bar{\gamma}_{2c}^{\sigma_c} (\omega) := \bar{\gamma}_2^{\sigma_c}( \omega; \alpha_c, \beta_c), \quad \bar{\gamma}_{1c}^{-\sigma_c}(\omega) := \bar{\gamma}_1^{-\sigma_c}( \omega; \alpha_c, \beta_c), \quad \sigma_c \in \mathcal{MV}(\alpha_c, \beta_c), \\
&\bar{\gamma}_{2c}^{\sigma_d} (\omega) := \bar{\gamma}_2^{\sigma_d}( \omega; \alpha_d, \beta_d), \quad \bar{\gamma}_{1d}^{-\sigma_d}(\omega) := \bar{\gamma}_1^{-\sigma_d}( \omega; \alpha_d, \beta_d), \quad \sigma_d \in \mathcal{MV}(\alpha_d, \beta_d), \\
&\qquad \alpha_d := 2 \pi - \alpha_a - \alpha_b - \alpha_c, \\
\end{aligned}
\end{equation}
for $\omega \in [-\pi, \pi]$ and for the folding angle functions in Theorem \ref{VertexTheorem}.  Using these definitions, we obtain the following characterization.
\begin{theorem}\label{PanelTheorem}
Consider the crease pattern of a single panel, as depicted in Fig.\;\ref{fig:SinglePanel}(b).  Assume every vertex of this pattern satisfies Kawasaki's condition  with sector angles as in (\ref{eq:sectorPanel}).  Then the crease pattern admits non-degenerate\footnote{During the process of folding, the folding angles at each vertex must satisfy the constraint (\ref{eq:nonDegeneracy}) away from the flat and folded flat state.  That is, we will not be considering the degenerate cases involving folding-in-half along a crease outlining one of the sides of the panel.} RFFQM if and only if 
\begin{equation}
\begin{aligned}\label{eq:compatPanel}
&\bar{\gamma}_{1c}^{-\sigma_c} \circ \bar{\gamma}_{2d}^{\sigma_d} \circ \bar{\gamma}_{1b}^{-\sigma_b} \circ \bar{\gamma}_{2a}^{\sigma_a} = id  \quad \text{ for some } \quad   (\sigma_a, \sigma_b, \sigma_c, \sigma_d) \in \mathcal{MV}_{abcd},
\end{aligned}
\end{equation}
where $\mathcal{MV}_{abcd} := \mathcal{MV}(\alpha_a, \beta_a) \times \mathcal{MV}(\alpha_b, \beta_b) \times \mathcal{MV}(\alpha_c, \beta_c) \times \mathcal{MV}(2\pi - \alpha_a - \alpha_b - \alpha_c, \beta_d)$.
\end{theorem} 
\noindent Here, we have employed the composition of maps notation, e.g., $\bar{\gamma}_{1b}^{-\sigma_b} \circ \bar{\gamma}_{2a}^{\sigma_a}(\omega) := \bar{\gamma}_{1b}^{-\sigma_b}( \bar{\gamma}_{2a}^{\sigma_a}(\omega))$.  

Let us comment briefly on some important aspects of this result  before turning to its derivation.  In the \textit{panel compatibility condition} (\ref{eq:compatPanel}), the left-hand-side is a composition of four folding angle functions that depend on quantities at an individual vertex: These functions depend on their vertex sector angles  $(\alpha_{(\cdot)}, \beta_{(\cdot)})$ in a highly non-linear way, and their structure is indicated by a sign $\sigma_{(\cdot)}$ corresponding to the vertex mountain-valley assignment.   As a result, the overall composition is a highly non-linear function of all the sector angles which make up the panel's crease pattern, and all its possible mountain-valley assignments.  So, satisfying (\ref{eq:compatPanel})  without simplifying assumptions on the parameters is non-trivial. This will be our task in the coming section. 


 For now, though, we  highlight another observation potentially relevant to applications: While it will turn out that most RFFQM crease patterns are foldable in exactly one way, this is by no means dictated by the theorem at hand.  For a given crease pattern, i.e., a given collection of sector angles $\alpha_a, \ldots, \beta_d$, it may turn out that multiple collections of signs $(\sigma_a, \sigma_b, \sigma_c, \sigma_d), (\tilde{\sigma}_a, \tilde{\sigma}_b, \tilde{\sigma}_c, \tilde{\sigma}_d), \ldots \in \mathcal{MV}_{abcd}$ can be used to establish the result (\ref{eq:compatPanel}).  In such cases, the panel's crease pattern can be folded from flat to folded flat along multiple distinct mountain-valley assignments, each corresponding to the signs $(\sigma_a, \sigma_b, \sigma_c, \sigma_d), (\tilde{\sigma}_a, \tilde{\sigma}_b, \tilde{\sigma}_c, \tilde{\sigma}_d), \ldots$ achieving panel compatibility.  We will discuss this multi-stability only in passing here.   We refer the interested reader to \cite{dieleman2020jigsaw} for a similar observation and some development on its applications. 
 
We now derive panel compatibility in Theorem \ref{PanelTheorem}.  We first isolate a single panel and its surrounding crease pattern (Fig.\;\ref{fig:SinglePanel}(a-b)).  On this panel, we assign crease tangents $\mathbf{t}_1^{a}, \ldots, \mathbf{t}_4^{a} \in \mathbb{S}^2$ to the $a$-vertex, $\mathbf{t}_1^b, \ldots,  \mathbf{t}_4^{b} \in \mathbb{S}^2$ to the $b$-vertex, $\mathbf{t}_1^{c}, \ldots, \mathbf{t}_4^{c} \in \mathbb{S}^2$ to the $c$-vertex, and $\mathbf{t}_1^d, \ldots,  \mathbf{t}_4^{d} \in \mathbb{S}^2$ to the $d$-vertex.  As indicated with Fig\;\ref{fig:SinglePanel}(b),  the crease assignments for the $a$ and $d$-vertex have a counter-clockwise labeling and the assignments for the $b$ and $c$-vertex have  a clockwise labeling.  

If the panel is to be folded as non-degenerate rigid origami, then each individual vertex must be folded as non-degenerate rigid origami. We therefore assign the deformation gradients surrounding each vertex as in Fig.\;\ref{fig:SinglePanel}(c).  Considering the notation, the $a$ and $d$-vertex assignments are obtained by direct application of Theorem \ref{VertexTheorem} using the vertex deformation (\ref{eq:localDef}).  Alternatively, since the tangents of the $b$ and $c$-vertex are labeled clockwise, we simply apply the theorem but replace all the folding angles with their minus.  This has the effect of keeping a consistent folding angle domain (\ref{eq:MVAngles}) associated to mountain and valley creases.  

Deriving (\ref{eq:compatPanel}) is then simply a question of consistency of the deformation gradients surrounding the panel (equivalently, consistency of the folding angles).   As the rotations $\mathbf{R}_{1,\ldots,4}^{a,\ldots,d}(\cdot)$ are $2\pi$ periodic, we take the primary branch for the folding angle domains without loss of generality. Then, this consistency is equivalent to 
\begin{equation}
\begin{aligned}\label{eq:forConsistency}
\omega_a = \omega_c, \quad \omega_d = \omega_b, \quad \bar{\gamma}_{2a}^{\sigma_a}(\omega_a) = \bar{\gamma}_{2b}^{\sigma_b}(\omega_b), \quad \bar{\gamma}_{2c}^{\sigma_c}(\omega_c) = \bar{\gamma}_{2d}^{\sigma_d}(\omega_d).
\end{aligned}
\end{equation}
For instance, the $a$-vertex has a rotation $\mathbf{R}_1^a(-\omega_a)$ on the same panel for which the $c$-vertex has rotation $\mathbf{R}_{1}^c(\omega_c)$.  The two rotations must therefore be the same; thus,  yielding the condition on the angles $\omega_a = \omega_c$ since $\mathbf{t}_1^{c} = -\mathbf{t}_1^{a}$. There are three analogous consistency conditions, so this reasoning leads to the four total conditions on the folding angle surrounding the panel in (\ref{eq:forConsistency}).  Panel compatibility (\ref{eq:compatPanel}) then emerges as a requirement  for these consistency conditions:  Recall that the inverse of $\bar{\gamma}_2^{\sigma}$ for $\sigma \in \mathcal{MV}(\alpha,\beta)$ is $\bar{\gamma}_1^{-\sigma}$. We can therefore substitute the first two equations in (\ref{eq:forConsistency}) into the latter two, apply the inverse $\bar{\gamma}_{1b}^{-\sigma_b}$ to the third and $\bar{\gamma}_{1c}^{-\sigma_c}$ to the fourth, and substitute the third equation into the fourth to arrive at the requirement 
\begin{equation}
\begin{aligned}\label{eq:prePanel}
\bar{\gamma}_{1c}^{-\sigma_c} \circ \bar{\gamma}_{2d}^{\sigma_d} \circ \bar{\gamma}_{1b}^{-\sigma_b} \circ \bar{\gamma}_{2a}^{\sigma_a} (\omega_a) = \omega_a.
\end{aligned}
\end{equation}

Now let us assume this requirement holds for some collection of sector angels $\alpha_a,\ldots, \beta_d$ in (\ref{eq:sectorPanel}), some choice of mountain-valley assignment $(\sigma_a, \ldots, \sigma_d) \in \mathcal{MV}_{abcd}$, and some folding parameter $\omega_a$.  It is then possible to solve (\ref{eq:forConsistency}).  Further, by substituting this solution into the deformation gradients in Fig.\;\ref{fig:SinglePanel}(c),  we obtain a rigid folding of the panel's crease pattern with a folding angle $\omega_a$ at the $\mathbf{t}_1^a$ crease, as indicated schematically with Fig.\;\ref{fig:SinglePanel}(d).  Note though, we are interested in RFFQM. So not only should (\ref{eq:prePanel}) hold for some $\omega_a$, it should hold \textit{for all} $\omega_a \in [-\pi, \pi]$.  In other words, this composition of maps should be the identity, i.e., yield $\omega_a$ when evaluated at $\omega_a$ regardless of the choice of this angle.  Consequently, we obtain panel compatibility (\ref{eq:compatPanel}) as the necessary and sufficient condition for non-degenerate RFFQM about this panel.

\section{Rigidity of the global kinematics}\label{sec:rigidSection}

A classification of all designs and deformations of RFFQM emerges from a rigidity result on the solutions to panel compatibility.    The key idea is to recognize that the folding angle functions and their compositions\textemdash however complicated is their sector angle and mountain-valley dependence\textemdash share an underlying group structure.  This group structure then reduces the difficult problem of characterizing these compositions to a much simpler problem of characterizing a corresponding multiplication of parameters.  From this reduction, we are able to obtain all solutions to panel compatibility.  The structure of these solutions, in turn, yields a concrete algorithm for the designs and deformations of all RFFQM.


\subsection{An Abelian group structure to folding angle functions}

In an effort to have a general classification of functions which ``look like" the folding angle functions defined in Theorem \ref{VertexTheorem}, we study the functions
\begin{equation}
\begin{aligned}\label{eq:gDef}
g(\omega; \mu) := -\sign \big(\mu \omega \big) \arccos\Big( \frac{(\mu^2 + 1) \cos \omega  + (\mu^2 - 1)}{(\mu^2 + 1) + (\mu^2 - 1) \cos \omega }  \Big), \quad \mu \in \mathbb{R} \setminus \{ 0\}.
\end{aligned}
\end{equation}
We first observe that any folding angle function can be written in this form for the suitable choice of $\mu \in \mathbb{R} \setminus \{ 0\}$. 
\begin{proposition}\label{firstProp}
Let $\alpha, \beta \in (0, \pi)$, $(\alpha, \beta) \neq (\pi/2, \pi/2)$, and $\sigma \in \mathcal{MV}(\alpha, \beta)$.  We have
\begin{equation}
\begin{aligned}\label{eq:gamma2Gfold}
\bar{\gamma}_2^{\sigma}(\omega; \alpha, \beta) = g\big( \omega; \tfrac{-\sigma 1 + \cos \alpha \cos \beta + \sin \alpha \sin \beta}{ \cos \beta - \sigma \cos \alpha} \big), \quad \bar{\gamma}_1^{-\sigma}(\omega; \alpha , \beta) = g \Big( \omega ;  \tfrac{\sigma 1 - \cos \alpha \cos \beta + \sin \alpha \sin \beta}{ -\cos \beta + \sigma \cos \alpha} \Big) .
\end{aligned}
\end{equation}
\end{proposition}
\begin{proof}
We only need to prove the result for $\bar{\gamma}_2^{\sigma}$ since $\bar{\gamma}_1^{-\sigma}$ is obtained from the identity in (\ref{eq:inverse}). In this direction, we notice that $\cos \alpha - \sigma \cos \beta \neq 0$ since $\sigma \in \mathcal{MV}(\alpha, \beta)$.  Therefore, we may define the parameters $\eta :=\tfrac{-\sigma 1 + \cos \alpha \cos \beta}{\cos \beta - \sigma \cos \alpha}$ and $\delta: = \tfrac{\sin \alpha \sin \beta}{\cos \beta - \sigma \cos \alpha}$, and introduce the parameterization 
\begin{equation}
\begin{aligned}
h(\omega; \eta, \delta) :=  -\sign\big( (\eta + \delta ) \omega  \big) \arccos \Big(\frac{\eta \cos \omega + \delta}{ \eta + \delta \cos \omega}\Big).
\end{aligned}
\end{equation} 
We claim that $\bar{\gamma}_2^{\sigma}(\omega; \alpha, \beta) = h(\omega; \eta, \delta)$.  Indeed, it is clear that the $\arccos(\cdot)$ part of this function is the same as that in $\bar{\gamma}_2^{\sigma}$; so we just need to show that $\sign( \eta + \delta) = \sign ( \cos \alpha - \sigma \cos \beta)$.  We observe that 
\begin{equation}
\begin{aligned}
\sign (\eta + \delta) = \sign\Big(\frac{-\sigma 1 + \cos \alpha \cos \beta + \sin \alpha \sin \beta }{\cos \beta - \sigma \cos \alpha} \Big) =   \sign\Big(\frac{-\sigma 1 }{\cos \beta - \sigma \cos \alpha} \Big),
\end{aligned}
\end{equation}
which yields the desired result after rearranging the latter.  Thus, we have $\bar{\gamma}_2^{\sigma}(\omega; \alpha, \beta) = h(\omega; \eta, \delta)$. 


Next, we observe that $\eta^2 - \delta^2 = 1$  by a direct calculation.  We therefore define $\mu := \eta + \delta$, and it follows that $\eta = \tfrac{\mu^2 + 1}{2\mu}$ and $\delta = \frac{\mu^2-1}{2\mu}$ using the aforementioned identity.  By expanding terms, we have
\begin{equation}
\begin{aligned}
\mu = \frac{-\sigma 1 + \cos \alpha \cos \beta + \sin \alpha \sin \beta}{ \cos \beta - \sigma \cos \alpha}.
\end{aligned}
\end{equation}
Finally, in combining all of these result, we  obtain the chain of equalities
\begin{equation}
\begin{aligned}
\bar{\gamma}_2^{\sigma}(\omega; \alpha, \beta) = h\Big(\omega; \frac{\mu^2 + 1}{2\mu}, \frac{\mu^2 - 1}{2\mu}\Big)  = g( \omega ; \mu).
\end{aligned}
\end{equation}
This completes the proof. 
\end{proof}

Next, for definiteness, we have: 
\begin{proposition}
The function $g(\omega; \mu)$ is a well-defined mapping from $[-\pi, \pi]$ to $[-\pi, \pi]$  for any $\mu \in \mathbb{R} \setminus \{ 0 \}$.  
\end{proposition} 
\begin{proof}
We simply have to verify that the argument in the $\arccos( \cdot)$ is well-defined (i.e., not dividing by zero) and has magnitude $\leq 1$.  In this direction, we observe that 
\begin{equation}
\begin{aligned}
&\big( (\mu^2 + 1) + (\mu^2 - 1) \cos \omega\big)^2  - \big( (\mu^2 + 1) \cos \omega + (\mu^2 - 1) \big)^2 \\
&\qquad \qquad = ((\mu^2+ 1)^2 - (\mu^2 - 1)^2)(1- \cos^2 \omega)  = 4 \mu^2 ( 1 - \cos^2\omega)  \geq 0.
\end{aligned}
\end{equation}
This establishes that the argument has magnitude $\leq 1$, or the argument is $0/0$.  The latter is, however, not a possibility.  Notice that the inequality above is strict unless $\omega= 0$ since $\mu \neq 0$.  However, also notice that $(\mu^2 + 1) + (\mu^2 - 1) \cos \omega = 2 \mu^2 \neq 0$ when $\omega =0$.  Thus, we never divide by zero in this parameterization.   
\end{proof}

In light of the above definiteness of the functions in (\ref{eq:gDef}), we can introduce without any ambiguity the collection of all such functions, i.e., 
\begin{equation}
\begin{aligned}\label{eq:GDef}
\mathcal{G} := \Big\{ g \colon [-\pi, \pi] \rightarrow [-\pi, \pi] \;\; \big| \;\;  g = g(\omega ; \mu ) \text{ as above, } \;\; \mu \in \mathbb{R} \setminus \{ 0\}  \Big\} .
\end{aligned}
\end{equation}
This collection has the following {remarkable} property:
\begin{lemma}\label{GroupLemma}
 $\mathcal{G}$ in (\ref{eq:GDef}) is an abelian group under the composition of maps product.  Specifically, it has the properties:
\begin{itemize}
\item The functions $g_1 := g(\omega;\mu_1)$, $g_2 := g(\omega; \mu_2) \in \mathcal{G}$  satisfy $g_1 = g_2$ if and only if $\mu_1 = \mu_2$.  
\item The identity map is $id = g(\omega; 1)$;  
\item  For any $g_1 := g(\omega; \mu_1)$ and $g_2(\omega; \mu_2) \in \mathcal{G}$, the group product satisfies 
\begin{equation}
\begin{aligned}\label{eq:productRule}
g_1 \circ g_2 = g_2 \circ g_1 = g(\omega, \mu_1 \mu_2) \in \mathcal{G};
\end{aligned}
\end{equation}
\item  The inverse of  $g_0 := g(\omega, \mu_0) \in \mathcal{G}$ is given by $g_0^{-1} = g(\omega, \mu_0^{-1}) \in \mathcal{G}$.
\end{itemize}
\end{lemma}
\begin{proof}
Composition of functions is always associative, so establishing the four properties above verifies that $\mathcal{G}$ is an abelian group and that its parameterization by $\mu \in \mathbb{R} \setminus \{0\}$ has no repeated elements.  For the first property, we notice that $g_1 = g_2$ only if 
\begin{equation}
\begin{aligned}
\frac{(\mu_1^2 + 1) \cos \omega + (\mu_1^2-1)}{(\mu_1^2 + 1) \cos \omega + (\mu_1^2-1)} = \frac{(\mu_2^2 + 1) \cos \omega + (\mu_2^2-1)}{(\mu_2^2 + 1) \cos \omega + (\mu_2^2-1)} \quad \forall \;\;  \omega \in [-\pi, \pi]
\end{aligned}
\end{equation}
since $\arccos(\cdot)$ is monotonically decreasing function of $[-1,1]$.  By rearranging terms, this equation becomes the condition
\begin{equation}
\begin{aligned}
&\big(  \mu_2^2 - \mu_1^2 \big) \cos \omega + \mu_1^2 - \mu_2^2 = 0, \quad \forall \;\; \omega \in [-\pi, \pi].
\end{aligned}
\end{equation}
 Clearly then, we require $\mu_1 = \pm \mu_2$ since $\cos \omega$ is a smooth non-constant function on this interval.   Taking into account the $\sign(\cdot)$ part of $g(\omega;\mu)$, we obtain $\mu_1 = \mu_2$ as desired.
 For the second property, we notice $g(\omega; 1) = \sign (\omega) \arccos(\cos(\omega)) = \sign (\omega) |\omega| = \omega$. For the third, we observe that 
\begin{equation}
\begin{aligned}
\big(g_1 \circ g_2\big) (\omega) &= g(g(\omega; \mu_2); \mu_1) \\
&= \sign(\mu_1\mu_2 \omega) \arccos \bigg( \frac{(\mu_1^2 + 1)\Big(\frac{(\mu_2^2 +1) \cos \omega + (\mu_2^2 - 1)}{(\mu_2^2 +1) + (\mu_2^2 -1) \cos \omega}\Big) + (\mu_1^2 - 1)}{(\mu_1^2 + 1)+ (\mu_1^2 - 1) \Big(\frac{(\mu_2^2 +1) \cos \omega + (\mu_2^2 - 1)}{(\mu_2^2 +1) + (\mu_2^2 -1) \cos \omega}\Big) } \bigg)   \\
& = \sign(\mu_1\mu_2 \omega) \arccos \bigg( \tfrac{(\mu_1^2 + 1)\big((\mu_2^2 +1) \cos \omega + (\mu_2^2 - 1)\big) + \big((\mu_2^2 +1) + (\mu_2^2 -1) \cos \omega\big)(\mu_1^2 - 1)}{(\mu_1^2 + 1) \big((\mu_2^2 +1) + (\mu_2^2 -1) \cos \omega\big)+ (\mu_1^2 - 1)\big( (\mu_2^2 +1) \cos \omega + (\mu_2^2 - 1)\big)} \bigg)  \\
& = \sign(\mu_1\mu_2 \omega) \arccos \bigg( \frac{(\mu_1^2 \mu_2^2 + 1) \cos(\omega) + (\mu_1^2\mu_2^2  -1)}{(\mu_1^2 \mu_2^2 + 1) + (\mu_1^2\mu_2^2  -1)  \cos(\omega) }\bigg),
\end{aligned}
\end{equation}
where the last identity is obtained simply by expanding out the terms in the second to last identity.  Notice this latter identity is nothing but $g(\omega; \mu_1 \mu_2)$.  Thus, $g_1 \circ g_2 = g(\omega; \mu_1 \mu_2) = g(\omega; \mu_2 \mu_1)  = g_2 \circ g_1$, as desired.  Finally, the last property  is a direct consequence of the other two.
\end{proof} 

\subsection{A design theorem for rigidly and flat-foldable origami.}

We now return to the study of  panel compatibility in (\ref{eq:compatPanel}).  Our brief digression regarding the Abelian group\textemdash particularly, our observation in (\ref{eq:gamma2Gfold}) concerning the folding angle functions\textemdash has proven useful in the following sense: It beckons the introduction of the sector angle and mountain-valley dependent parameters 
\begin{equation}
\begin{aligned}\label{eq:foldAngMult}
\mu_2^{\sigma}(\alpha, \beta) :=  \frac{-\sigma 1 + \cos \alpha \cos \beta + \sin \alpha \sin \beta}{ \cos \beta - \sigma \cos \alpha}, \quad \mu_1^{-\sigma} ( \alpha, \beta) := \mu_2^{-\sigma}(\alpha, \pi - \beta) , \quad \sigma \in \mathcal{MV}(\alpha,\beta).
\end{aligned}
\end{equation}  
In keeping with the origami literature \cite{evans2015rigidly, lang2011origami, tachi2009generalization, tachi2010geometric}, we call these parameters the \textit{fold angle multipliers}.  We should point out that, although these parameters are known to this literature, here we make the direct connection to the kinematics (i.e., (\ref{eq:gamma2Gfold}) to (\ref{eq:foldAngMult})) through an Abelian group transparent and rigorous.

A key step in our characterization of panel compatibility is relating the composition of folding angle functions directly to a product of fold angle multipliers.  In this direction, we introduce the vertex-dependent fold angle multipliers
\begin{equation}
\begin{aligned}\label{eq:multipliersVertex}
&\mu_{2a}^{\sigma} := \mu_2^{\sigma}(\alpha_a, \beta_a), \quad \mu_{1a}^{-\sigma} := \mu_1^{-\sigma}(\alpha_a, \beta_a), \quad \sigma \in \mathcal{MV}(\alpha_a, \beta_a),  \\
&\mu_{2b}^{\sigma} := \mu_2^{\sigma}(\alpha_b, \beta_b), \quad \mu_{1b}^{-\sigma} := \mu_1^{-\sigma}(\alpha_b, \beta_b), \quad \sigma \in \mathcal{MV}(\alpha_b, \beta_b),  \\
&\mu_{2c}^{\sigma} := \mu_2^{\sigma}(\alpha_c, \beta_c), \quad \mu_{1c}^{-\sigma} := \mu_1^{-\sigma}(\alpha_c, \beta_c), \quad \sigma \in \mathcal{MV}(\alpha_c, \beta_c),  \\
&\mu_{2d}^{\sigma} := \mu_2^{\sigma}(\alpha_d, \beta_d), \quad \mu_{1d}^{-\sigma} := \mu_1^{-\sigma}(\alpha_d, \beta_d), \quad \sigma \in \mathcal{MV}(\alpha_d, \beta_d), \\
&\qquad \alpha_d := 2 \pi - \alpha_a - \alpha_b - \alpha_c,
\end{aligned}
\end{equation}
for admissible collections of sector angles $\alpha_a,\ldots,\beta_d$ as in (\ref{eq:sectorPanel}). The connection between panel compatibility and the fold angle multipliers can now be made precise.
\begin{lemma}\label{EquivLemma}
For any collection of sector angle $\alpha_a, \ldots, \beta_d$ in (\ref{eq:sectorPanel}) and any mountain-valley assignment $(\sigma_a, \sigma_b, \sigma_c, \sigma_d) \in \mathcal{MV}_{abcd}$ (recall Theorem \ref{PanelTheorem} for the definition), the following equivalence holds:
\begin{equation}
\begin{aligned}\label{eq:equivalenceImport}
\bar{\gamma}_{1c}^{-\sigma_c} \circ \bar{\gamma}_{2d}^{\sigma_d} \circ  \bar{\gamma}_{1b}^{-\sigma_b} \circ \bar{\gamma}_{2a}^{\sigma_a} = id \quad \Leftrightarrow \quad \mu_{1c}^{-\sigma_c} \mu_{2d}^{\sigma_d}\mu_{1b}^{-\sigma_b} \mu_{2a}^{\sigma_a}  = 1.
\end{aligned}
\end{equation}
\end{lemma}
\begin{proof}
Since $\bar{\gamma}_{1c}^{-\sigma_c} \circ \bar{\gamma}_{2d}^{\sigma_d} \circ  \bar{\gamma}_{1b}^{-\sigma_b} \circ \bar{\gamma}_{2a}^{\sigma_a}$ is the composition of functions that belong to the group $\mathcal{G}$ (Proposition \ref{firstProp} and Lemma \ref{GroupLemma}), it follows that $\bar{\gamma}_{1c}^{-\sigma_c} \circ \bar{\gamma}_{2d}^{\sigma_d} \circ  \bar{\gamma}_{1b}^{-\sigma_b} \circ \bar{\gamma}_{2a}^{\sigma_a} \in \mathcal{G}$.  Consequently, there exists a $\mu_{abcd} \in \mathbb{R} \setminus \{ 0\}$ such that 
\begin{equation}
\begin{aligned}
\bar{\gamma}_{1c}^{-\sigma_c} \circ \bar{\gamma}_{2d}^{\sigma_d} \circ  \bar{\gamma}_{1b}^{-\sigma_b} \circ \bar{\gamma}_{2a}^{\sigma_a}(\omega) = g(\omega; \mu_{abcd}), \quad \omega \in [-\pi, \pi].  
\end{aligned}
\end{equation}
By combining the formula (\ref{eq:gamma2Gfold}) with the product rule (\ref{eq:productRule}), we obtain $\mu_{abcd} =  \mu_{1c}^{-\sigma_c} \mu_{2d}^{\sigma_d}\mu_{1b}^{-\sigma_b} \mu_{2a}^{\sigma_a}$.   Since $g(\omega; \mu) = id$ if and only if $\mu = 1$ (the first two properties in Lemma \ref{GroupLemma}), we conclude (\ref{eq:equivalenceImport}) as desired. 
\end{proof}

With the equivalence in (\ref{eq:equivalenceImport}) now established, we are prepared to develop the main characterization theorem for RFFQM.  In this characterization, we consider an arbitrary crease pattern in Fig.\;\ref{fig:SinglePanel}(b) under Kawasaki's condition and treat quantities at three of the four vertices, i.e., sector angles and mountain-valley assignments, as given under the hypotheses
\begin{equation}
\begin{aligned}\label{eq:DesignHypotheses}
(\text{Design parameters:}) \quad \begin{cases}
\alpha_a, \alpha_b, \alpha_c \in (0,\pi), \quad 2\pi -\alpha_a  - \alpha_b - \alpha_c \in (0,\pi), \quad \beta_a, \beta_b , \beta_c \in (0,\pi), \\
(\alpha_a, \beta_a), (\alpha_b, \beta_b), (\alpha_c, \beta_c) \neq (\pi/2, \pi/2),\\
\sigma_a \in \mathcal{MV}(\alpha_a, \beta_a), \;\; \sigma_b \in \mathcal{MV}(\alpha_b, \beta_b), \;\; \sigma_c \in \mathcal{MV}(\alpha_c, \beta_c).
\end{cases}
\end{aligned}
\end{equation} 
We then ask the fundamental design questions: Can the final sector angle $\beta_d \in (0,\pi)$ and mountain-valley assignment $\sigma_d \in \mathcal{MV}(2\pi-\alpha_a - \alpha_b - \alpha_c, \beta_d)$  be chosen to yield  a crease pattern that admits RFFQM. If so, how rigid or flexible is the design criterion?  To this end, we obtain an explicit rigidity theorem for design.
\begin{theorem}\label{DesignTheorem}
Consider a crease pattern as in Fig\;\ref{fig:SinglePanel}(b) under Kawasaki's condition, and let the sector angles $\alpha_a, \ldots, \beta_c$ and mountain-valley assignments $\sigma_a, \sigma_b,\sigma_c$ at the $a$-$b$-$c$ vertices be given under the hypotheses in (\ref{eq:DesignHypotheses}).  The characterization of the $d$-vertex is based on the evaluation of 
\begin{equation}
\begin{aligned}\label{eq:muABC}
\mu_{abc} := \mu_{1a}^{-\sigma_a} \mu_{2b}^{\sigma_b} \mu_{2c}^{\sigma_c}.  
\end{aligned}
\end{equation}
\begin{itemize}
\item If $|\mu_{abc}| = 1$, then there is no pair $(\beta_d, \sigma_d) \in  (0,\pi) \times \mathcal{MV}(2\pi-\alpha_a - \alpha_b - \alpha_c, \beta_d)$ yielding a crease pattern that can be folded as non-degenerate\footnote{See the footnote on this topic in Theorem \ref{PanelTheorem}.} RFFQM along the mountain-valley assignments $\sigma_a, \sigma_b, \sigma_c$ at the $a$-$b$-$c$ vertices.
\item Otherwise, $|\mu_{abc}| \neq 1$ and there is a unique pair $(\beta_d, \sigma_d) \in  (0,\pi) \times \mathcal{MV}(2\pi-\alpha_a - \alpha_b - \alpha_c, \beta_d)$ yielding a crease pattern that can be folded as non-degenerate RFFQM along the mountain-valley assignments $\sigma_a, \sigma_b, \sigma_c$ at the $a$-$b$-$c$ vertices. This pair satisfies the explicit relationship
\begin{equation}
\begin{aligned}\label{eq:SuperParam}
\sigma_d1= -\sign\big( \mu_{abc}^2 -1  \big), \quad \beta_d = \arccos\Big(\sigma_d \frac{2 \mu_{abc} - (\mu_{abc}^2 + 1) \cos(\alpha_a + \alpha_b + \alpha_c) }{2\mu_{abc} \cos(\alpha_a + \alpha_b + \alpha_c) - (\mu_{abc}^2 + 1)} \Big).
\end{aligned}
\end{equation}
\end{itemize}
\end{theorem}

A couple of points with this theorem:  The case $|\mu_{abc}| =1$ is highly non-generic.  In fact, one can treat $\alpha_a, \alpha_b, \alpha_c, \beta_b, \beta_c, \sigma_b, \sigma_c$ as given (under the hypotheses (\ref{eq:DesignHypotheses})) and use an argument similar to the proof of this theorem to find that there are at most two pairs $(\beta_a, \sigma_a) \in (0,\pi) \times \mathcal{MV}(\alpha_a, \beta_a)$ which give this degeneracy.  As a result, a generic set of design parameters (\ref{eq:DesignHypotheses}) at the $a$-$b$-$c$ vertices will admit RFFQM by the prescription in (\ref{eq:SuperParam}).  A corollary to this effect is stated in the next subsection.  A second point is that this result, as stated, is very much mountain-valley assignment dependent:  We are testing whether a crease pattern can deform as RFFQM \textit{along} specific mountain-valley assignments $\sigma_a, \sigma_b, \sigma_c$ at the $a$-$b$-$c$ vertices.  Thus, when $|\mu_{abc}| = 1$, one should avoid the interpretation that the sector angles $\alpha_a, \ldots, \beta_c$ cannot admit a RFFQM crease pattern; only that they cannot admit one whose folding corresponds to mountain-valley assignments $\sigma_a, \sigma_b, \sigma_c$ at the aforementioned vertices.  

There are other results in similar spirit to Theorem \ref{DesignTheorem} in the literature.  Building on ideas for Tachi \cite{tachi2009generalization,tachi2010geometric}, Lang  and Howell \cite{lang2018rigidly} made the astute observation that one can solve panel compatibility generically by treating the folding angles as the independent variables rather than the sector angles.  Nevertheless, we find the sector angle to be the more natural design variables, especially for developing optimization schemes to tackle the inverse problem (see the discussion in Section \ref{ssec:Inverse}).  We should also point out the results of Izmestiev \cite{izmestiev2017classification} along these lines.  He classified all ``Kokotsakis polyhedra"\textemdash single panel crease patterns where no assumption are made on the sector angles\textemdash by polynomializing general constraints of the type in (\ref{eq:compatPanel}) and studying their zeros in the complex plane. While certainly a mathematical triumph, this characterization is  unlike Theorem \ref{DesignTheorem} in that it does not lend itself naturally to a design principle for engineering.  For instance, most of the solutions, beyond the variety we discuss for RFFQM, involve delicate couplings between all the vertices of the panel, making it hard to elucidate a procedure to go from a single panel to an overall crease pattern.  Some efforts to address these issues can be found in \cite{he2018rigid}, but these efforts are restricted to special cases.

\subsection{Derivation of the design theorem}

In this section, we prove Theorem \ref{DesignTheorem}.  We also state a corollary that justifies that the subcase $|\mu_{abc}| =1$ in this theorem is non-generic.  

\begin{proof}[Proof of Theorem \ref{DesignTheorem}.]
From Theorem \ref{PanelTheorem} and Lemma \ref{EquivLemma}, characterizing RFFQM patterns, as stated in the theorem, is equivalent to characterizing the pairs $(\beta_d, \sigma_d) \in (0, \pi) \times \mathcal{MV}(2\pi -\alpha_a - \alpha_b - \alpha_c, \beta_d)$ such that  $(\mu_{1c}^{-\sigma_c} )(\mu_{2d}^{\sigma_d} )(\mu_{1b}^{-\sigma_b})( \mu_{2a}^{\sigma_a}) = 1$.  Since $\bar{\gamma}_{1a}^{-\sigma_a}$ is the inverse of $\bar{\gamma}_{2a}^{\sigma_a}$, we observe that $\mu_{1a}^{-\sigma_a} = (\mu_{2a}^{\sigma_a})^{-1}$ due to the inverse property of group elements in Lemma \ref{GroupLemma}.  Analogous observations hold for all the other vertices.  We can therefore rearrange the equation $(\mu_{1c}^{-\sigma_c} )(\mu_{2d}^{\sigma_d} )(\mu_{1b}^{-\sigma_b})( \mu_{2a}^{\sigma_a}) = 1$ as $\mu_{2d}^{\sigma_d} = \mu_{abc}$ for $\mu_{abc}$ given in (\ref{eq:muABC}).  By Proposition \ref{PropDesign} below, we observe that 
\begin{equation}
\begin{aligned}\label{eq:importEquiv}
\mu_{2d}^{\sigma_d} = \mu_{abc} \quad \Leftrightarrow \quad \frac{(\mu_{2d}^{\sigma_d})^2 + 1}{2\mu_{2d}^{\sigma_d}} = \frac{\mu_{abc}^2 + 1}{2\mu_{abc}} , \quad \sign\Big( \frac{(\mu_{2d}^{\sigma_d})^2 - 1}{2\mu_{2d}^{\sigma_d}}\Big) = \sign\Big( \frac{\mu_{abc}^2 - 1}{2 \mu_{abc}}\Big)
\end{aligned}
\end{equation}
since $\mu_{abc} \in \mathbb{R} \setminus \{ 0\}$ by assumption of the given parameters.  

We focus on solving the latter equivalence in terms of $(\beta_d , \sigma_d) \in (0, \pi) \times \mathcal{MV}(2\pi -\alpha_a - \alpha_b - \alpha_c, \beta_d)$.  By a direct calculation, we obtain
\begin{equation}
\begin{aligned}
\frac{(\mu_{2d}^{\sigma_d})^2 + 1}{2\mu_{2d}^{\sigma_d}} = \frac{-\sigma_d 1 + \cos (\alpha_a + \alpha_b + \alpha_c) \cos \beta_d }{ \cos \beta_d - \sigma_d \cos( \alpha_a + \alpha_b + \alpha_c)}.
\end{aligned}
\end{equation}
Consequently, we can rearrange the first equation in the latter equivalence in (\ref{eq:importEquiv}) as 
\begin{equation}
\begin{aligned}\label{eq:cosBetaD}
\big(2\mu_{abc} \cos( \alpha_a + \alpha_b + \alpha_c )  - (\mu_{abc}^2 + 1)\big) \cos \beta_d = \sigma_d \big(2 \mu_{abc} - (\mu_{abc}^2 +1) \cos(\alpha_a + \alpha_b + \alpha_c) \big).
\end{aligned}
\end{equation}
Notice that 
\begin{equation}
\begin{aligned}
&\big(2\mu_{abc} \cos( \alpha_a + \alpha_b + \alpha_c )  - (\mu_{abc}^2 + 1)\big)^2 - \big(2\mu_{abc}   - (\mu_{abc}^2 + 1) \cos( \alpha_a + \alpha_b + \alpha_c )\big)^2  \\
&\qquad = (\mu_{abc}^2 - 1)^2 \sin(\alpha_a + \alpha_b + \alpha_c)^2 \geq 0 \quad \text{ and  } = 0 \text{ if and only if } |\mu_{abc}|  =1,
\end{aligned}
\end{equation}
since $\alpha_a + \alpha_b + \alpha_c \in (\pi ,2\pi)$, i.e., since $\sin (\alpha_a + \alpha_b + \alpha_c) \neq 0$.  Thus, if $|\mu_{abc}| = 1$, we divide through by  $2\mu_{abc} \cos( \alpha_a + \alpha_b + \alpha_c )  - (\mu_{abc}^2 + 1) \neq 0$ in (\ref{eq:cosBetaD}) to deduce that $\cos \beta_d$ is required to equal $-1$ or $1$.  However, this is not allowed since we demand that $\beta_d \in (0,\pi)$.  This shows that there are no solutions if $\mu_{abc} \in \{ \pm 1\}$, as asserted in the theorem.  Alternatively, if $|\mu_{abc}| \neq 1$, we divide through by $2\mu_{abc} \cos( \alpha_a + \alpha_b + \alpha_c )  - (\mu_{abc}^2 + 1) \neq 0$  in (\ref{eq:cosBetaD}), and obtain that the right-hand side, after the division, has magnitude strictly less than $1$.  Therefore, $\beta_d \in (0,\pi)$ must have the parametrization in (\ref{eq:SuperParam}) to solve (\ref{eq:cosBetaD}) in this case.  

In enforcing this parameterization, it remains only to solve the sign condition in (\ref{eq:importEquiv}) for $\sigma_d \in \mathcal{MV}(2\pi -\alpha_a - \alpha_b - \alpha_c, \beta_d)$.  To this end, we first claim that any $\sigma_d \in \{\pm\}$ which solves the sign condition in (\ref{eq:importEquiv}) will correspond to the set $\sigma_d \in  \mathcal{MV}(2\pi -\alpha_a - \alpha_b - \alpha_c, \beta_d)$.  Assuming the opposite, we have $\cos \beta_d - \sigma_d \cos  (\alpha_a + \alpha_b + \alpha_c) = 0$. Yet, the left-hand side is parameterized as 
\begin{equation}
\begin{aligned}\label{eq:calcSigmaD}
\cos \beta_d - \sigma_d \cos  (\alpha_a + \alpha_b + \alpha_c) = \sigma_d \Big(\frac{2\mu_{abc}(1- \cos^2(\alpha_a + \alpha_b + \alpha_c))}{2\mu_{abc} \cos(\alpha_a + \alpha_b + \alpha_c) - (\mu_{abc}^2 + 1)} \Big) ,
\end{aligned}
\end{equation} 
and neither $\mu_{abc}$ or $1- \cos^2(\alpha_a + \alpha_b + \alpha_c))$ are zero.  This is the desired contradiction; so we simply have to parameterize the sign condition in (\ref{eq:importEquiv}) to complete the proof.  Now, by an explicit calculation, this condition is equivalent to 
\begin{equation}
\begin{aligned}\label{eq:finalNeededIdent}
\sign \Big(\frac{-\sin (\alpha_a + \alpha_b + \alpha_c) \sin \beta_d}{ \cos \beta_d - \sigma_d \cos (\alpha_a + \alpha_b + \alpha_c)}\Big) = \sign\Big( \frac{\mu_{abc}^2 - 1}{2 \mu_{abc}}\Big)
\end{aligned}
\end{equation}
The numerator on the left side $-\sin(\alpha_a + \alpha_b + \alpha_c) \sin \beta_d$ is strictly positive since $\beta_d \in (0,\pi)$ and $\alpha_a + \alpha_b + \alpha_c \in (\pi, 2\pi)$; so it can be discarded.  We therefore find that 
\begin{equation}
\begin{aligned}\label{eq:manip1}
&\sign \Big(\frac{-\sin (\alpha_a + \alpha_b + \alpha_c) \sin \beta_d}{ \cos \beta_d - \sigma_d \cos (\alpha_a + \alpha_b + \alpha_c)} \Big) = \sign\Big(\cos \beta_d - \sigma_d \cos (\alpha_a + \alpha_b + \alpha_c)  \Big)  \\
&\qquad = \sigma_d  \sign\Big(\frac{1 - \cos(\alpha_a + \alpha_b + \alpha_c)^2}{\cos(\alpha_a + \alpha_b + \alpha_c) -\tfrac{(\mu_{abc}^2 + 1)}{2 \mu_{abc}}} \Big) 
\end{aligned}
\end{equation}
by routine manipulations and (\ref{eq:calcSigmaD}).  Finally, notice that $\cos(\alpha_a + \alpha_b + \alpha_c)^2< 1$ since $\alpha_a + \alpha_b + \alpha_c \in (\pi,2\pi)$ and$\left| \tfrac{(\mu_{abc}^2 + 1)}{2 \mu_{abc}}\right|  > 1 $ since $|\mu_{abc}| \neq 1$.  Thus, this result can be manipulated further; specifically, we have the identity
\begin{equation}
\begin{aligned}\label{eq:manip2}
\sigma_d  \sign\Big(\frac{1 - \cos(\alpha_a + \alpha_b + \alpha_c)^2}{\cos(\alpha_a + \alpha_b + \alpha_c) -\tfrac{(\mu_{abc}^2 + 1)}{2 \mu_{abc}}} \Big) = -\sigma_d \sign \Big( \frac{\mu_{abc}^2 + 1}{2\mu_{abc}} \Big).
\end{aligned}
\end{equation}
As a result of all of the manipulations in (\ref{eq:manip1}) and (\ref{eq:manip2}), the identity in (\ref{eq:finalNeededIdent}) demands that 
\begin{equation}
\begin{aligned}
\sigma_d  = - \sign \Big( \frac{\mu_{abc}^2 + 1}{2\mu_{abc}} \Big) \sign \Big(\frac{\mu_{abc}^2 - 1}{2\mu_{abc}} \Big) = -\sign( \mu_{abc}^2 -1).
\end{aligned}
\end{equation}
This is the desired result.  We have solved the equation in (\ref{eq:importEquiv}) for $|\mu_{abc}| \neq 1$ by the parameterizations in (\ref{eq:SuperParam}), and these parameterizations are clearly necessary.  
\end{proof}
\begin{proposition}\label{PropDesign}
Let $\check{\mu}, \hat{\mu} \in \mathbb{R} \setminus \{ 0\}$.  Then $\hat{\mu} = \check{\mu}$ if and only if 
\begin{equation}
\begin{aligned}\label{eq:interestIdents}
 \frac{\hat{\mu}^2 + 1}{2 \hat{\mu}} = \frac{\check{\mu}^2 + 1}{2 \check{\mu}} , \quad \sign\Big( \frac{\hat{\mu}^2 - 1}{2 \hat{\mu}}\Big) = \sign\Big( \frac{\check{\mu}^2 - 1}{2 \check{\mu}}\Big)
\end{aligned}
\end{equation}
\end{proposition}
\begin{proof}
Clearly, $\hat{\mu} = \check{\mu} \in \mathbb{R} \setminus \{ 0\}$ implies (\ref{eq:interestIdents}).  For the other direction, we first claim that (\ref{eq:interestIdents}) implies 
\begin{equation}
\begin{aligned}\label{eq:claimIdent}
 \frac{\hat{\mu}^2 - 1}{2 \hat{\mu}} = \frac{\check{\mu}^2 - 1}{2 \check{\mu}}.
\end{aligned}
\end{equation}
Indeed, as a direct consequence of the first equality and some algebraic manipulation, we have the chain of equalities 
\begin{equation}
\begin{aligned}
\Big(\frac{\hat{\mu}^2 -1}{2 \hat{\mu}}\Big)^2 = \Big(\frac{\hat{\mu}^2 +1} {2\hat{\mu}}\Big)^2  - 1 =  \Big(\frac{\check{\mu}^2 +1} {2\check{\mu}}\Big)^2  - 1 = \Big(\frac{\check{\mu}^2 -1}{2 \check{\mu}}\Big)^2.
\end{aligned}
\end{equation}
Here, the first and last equality are simply by rearrangement of terms and the second uses the first identity in (\ref{eq:interestIdents}).  Thus, with the second identity in (\ref{eq:interestIdents}), we obtain (\ref{eq:claimIdent}). We therefore have 
\begin{equation}
\begin{aligned}
\frac{\hat{\mu}}{2} = \frac{\hat{\mu}^2 + 1}{2 \hat{\mu}}  + \frac{\hat{\mu}^2 - 1}{2 \hat{\mu}}  = \frac{\check{\mu}^2 + 1}{2 \check{\mu}}  + \frac{\check{\mu}^2 - 1}{2 \check{\mu}}  = \frac{\check{\mu}}{2}
\end{aligned}
\end{equation}
using the first identity in (\ref{eq:interestIdents}) and the identity in (\ref{eq:claimIdent}). This completes the proof. 
\end{proof}

\begin{corollary}
Let $\alpha_a, \alpha_b, \alpha_c \in (0,\pi)$ with $2\pi - \alpha_a -\alpha_b - \alpha_c \in (0,\pi)$; let $\beta_b, \beta_c \in (0,\pi)$;  let $\sigma_b \in \mathcal{MV}(\alpha_b, \beta_b),\sigma_c \in \mathcal{MV}(\alpha_c, \beta_c)$; define $\mu_{bc} := \mu_{2b}^{\sigma_b} \mu_{2c}^{\sigma_c}$.   
\begin{itemize}
\item If $|\mu_{bc}| = 1$, then there are no pairs $(\beta_a, \sigma_a) \in (0,\pi) \times \mathcal{MV}(\alpha_a, \beta_a)$ such that $|\mu_{abc}| = 1$.  
\item Otherwise, $|\mu_{bc}| \neq 1$ and there are exactly two pairs $ (\beta^{\pm}_a, \sigma_a) \in (0,\pi) \times \mathcal{MV}(\alpha_a, \beta_a^{\pm})$ such that $|\mu_{abc}|  = 1$.  These pairs satisfy 
\begin{equation}
\begin{aligned}
\sigma_a1= -\sign\big( \mu_{bc}^2 -1  \big), \quad \beta^{\pm}_a = \arccos\Big(\sigma_a \frac{\pm 2 \mu_{bc} - (\mu_{bc}^2 + 1) \cos \alpha_a }{\pm2\mu_{bc} \cos\alpha_a  - (\mu_{bc}^2 + 1)} \Big).
\end{aligned}
\end{equation}
\end{itemize}
\end{corollary}

\section{Applications}\label{sec:Applications}
The design theorem for compatibility at a single panel (Theorem \ref{DesignTheorem}) lends itself naturally to engineering design principles for RFFQM. Particularly, we show how these results enable a \textit{marching algorithm} for explicitly and efficiently  exploring of the configuration space of \textit{all} RFFQM.  We then discuss how this algorithm can be incorporated into an optimization scheme to address the inverse problem: 
\begin{itemize}
\item Can we achieve a targeted three dimensional surface by the deployment of a RFFQM crease pattern?
\end{itemize}
Here we formulate the inverse problem and give some examples, mainly as a means to demonstrate the utility of our theoretical results.  In a forthcoming paper \cite{dang_inverse}, we address the inverse problem in a comprehensive manner, with a view towards broad engineering appeal and potential applications.


\begin{figure}
\centering
\includegraphics[width = 7in]{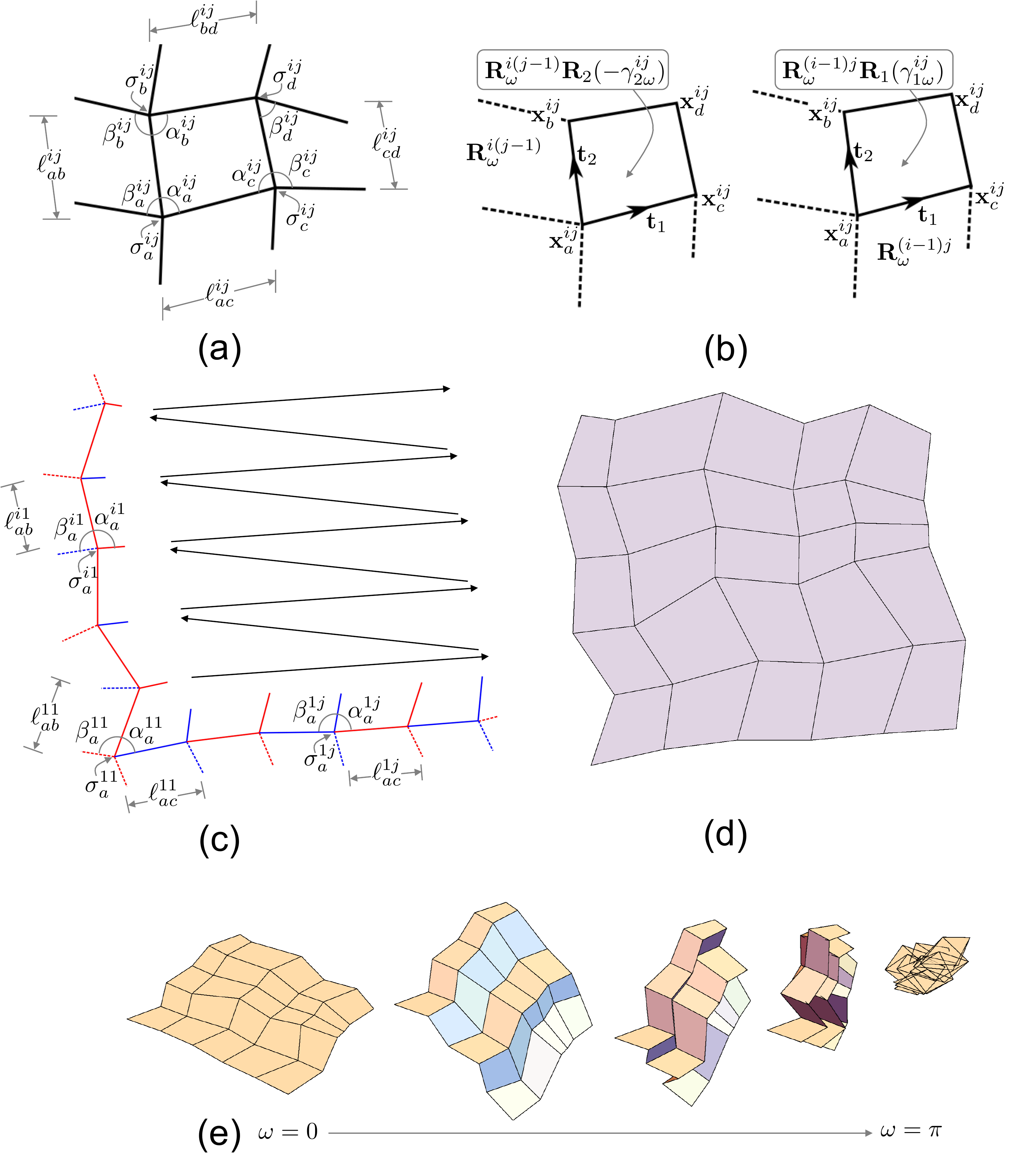}
\caption{A diagram describing the marching procedure.  (a) Notation for the geometry of the crease pattern at each panel.  (b) Notation for the deformation gradients at each panel.  (c). The input data for the marching algorithm used to obtain the crease pattern (d) and its kinematics (e). }
\label{fig:Marching}
\end{figure}

\subsection{A marching algorithm}
We introduce the marching algorithm by following the schematic in Fig.\;\ref{fig:Marching}.  Recall that the design theorem (Theorem \ref{DesignTheorem}) essentially states that, if we know the sector angles and mountain-valley assignments at three of four vertices of a panel, then for RFFQM of the panel's crease pattern: 1)\;there is a unique solution for the fourth vertex or 2)\;there is no solution.  With this in mind, suppose we are interested in a RFFQM crease pattern with $N$ rows of panels and $M$ columns of panels, so that the total number of panels for the quad-mesh crease pattern is $N\times M$.    We then have a design problem for a RFFQM crease pattern involving $N \times M$ applications of Theorem \ref{DesignTheorem}, which results in marching procedure to obtain the crease pattern and kinematics as a function of certain boundary sector angles and mountain-valley assignments.  Note, Lang and Howell \cite{lang2018rigidly} derived a formally similar marching algorithm, except that it involves treating folding angles, not sector angles, as the independent variables.  The advantages of a sector angle parameterization of marching will become apparent in our  discussion of the inverse problem (Section \ref{ssec:Inverse}).

A couple of additional points of emphasis are: 
\begin{itemize} 
\item The marching algorithm for the crease pattern and its kinematics is computationally efficient.  Precisely, the algorithm computes the design and kinematics of an $N \times M$ RFFQM crease pattern with a total computational time of
\begin{equation}
\begin{aligned}
\text{Computational Time} \sim O(N) \times O(M) \times O(\# \text{ of computed deformation}).
\end{aligned}
\end{equation}
Recall that each RFFQM has a one parameter family of origami deformations from flat to folded flat.  So the ``\# of computed deformations" above is simply the discretization of this one parameter family.
\item The marching algorithm for the crease pattern and its kinematics is also completely explicit.  In particular, every quantity of interest is obtained by a direct formula, not an optimization.   We emphasize this point below by explicitly enumerating all the formulas that are necessary and sufficient for the marching procedures (i.e., with (\ref{eq:marchingSector})-(\ref{eq:inputMarching}) for the crease pattern and (\ref{eq:creaseInputs})-(\ref{eq:finalCalc}) for the kinematics).  No detail, large or small, is suppressed.  As such, these formula can be directly implemented into any standard computational solver to obtain the designs and deformations of RFFQM. 
\end{itemize}

To explain the procedure, we will always suppose the scheme below marches from panel to panel using the ordering $(i,j) = (1,1), \ldots, (1,M), (2,1), \ldots , (N,M)$, consistent with the sketch in Fig.\;\ref{fig:Marching}(c). 

\bigskip

\textbf{Marching to obtain the crease pattern.} We populate the local sector angles, mountain-valley assignments  and side lengths for this design problem.  The notation here follows Fig.\;\ref{fig:Marching}(a). 

We consider the $(i,j)$-panel for $i, j > 1$, and let us assume that the ``marching" (by our to be determined scheme) has succeeded at every step prior to this panel.  Then, most of the sector angles, mountain-valley assignments and side lengths at the $(i,j)$-panel are informed by the previously determined panels. Specifically, the sector angles at this panel satisfy
\begin{equation}
\begin{aligned}\label{eq:marchingSector}
&\alpha_{a}^{ij} = \beta_c^{i(j-1)} , \quad \beta_a^{ij} = \alpha_{c}^{i(j-1)}, \quad \alpha_{b}^{ij} = \beta_d^{i(j-1)} , \quad \beta_b^{ij} = 2\pi - \alpha_{a}^{i(j-1)} - \alpha_b^{i(j-1)}  -\alpha_c^{i(j-1)} , \\
&\qquad \alpha_c^{ij} = \pi - \beta_d^{(i-1)j}, \quad  \beta_{c}^{ij} = -\pi+ \alpha_{a}^{(i-1)j} + \alpha_b^{(i-1)j}  +\alpha_c^{(i-1)j}
\end{aligned}
\end{equation}
due to Kawasaki's conditions, the mountain-valley assignments satisfy 
\begin{equation}
\begin{aligned}\label{eq:marchingMV}
\sigma_a^{ij} = \sigma_c^{i(j-1)}, \quad \sigma_b^{ij} = \sigma_d^{i(j-1)}, \quad \sigma_c^{ij} = \sigma_d^{(i-1)j},
\end{aligned}
\end{equation}
and the side length satisfy 
\begin{equation}
\begin{aligned}\label{eq:marchingSides}
\ell_{ab}^{ij} = \ell_{cd}^{i(j-1)}, \quad \ell_{ac}^{ij} = \ell_{bd}^{(i-1)j}.
\end{aligned}
\end{equation}
In contrast, the sector angle $\beta_d^{ij}$, mountain-valley assignment $\sigma_d^{ij}$, and the side lengths $\ell_{bd}^{ij}$ and $\ell_{cd}^{ij}$ are not determined from the panel's neighbors; instead, they are constrained by foldability and consistency considerations.   

Particularly, the pair $(\beta_d^{ij}, \sigma_d^{ij})$ at the panel is evaluated based on a direct application of Theorem \ref{DesignTheorem} to enforce a RFFQM crease pattern.  Explicitly, we check that 
\begin{equation}
\begin{aligned}\label{eq:firstCheck}
(\text{first check:}) \quad 2 \pi - \alpha^{ij}_a - \alpha^{ij}_b - \alpha^{ij}_c \in  (0,\pi),
\end{aligned}
\end{equation} 
as this is required for the interior angles of the $(i,j)$-panel sum to $2\pi$.  If successful, we then proceed to define the folding angle multiplier
\begin{equation}
\begin{aligned}
\mu_{abc}^{ij} :=  \mu_{1}^{-\sigma^{ij}_a}(\alpha_a^{ij}, \beta_a^{ij}) \mu_{2}^{\sigma^{ij}_b}(\alpha_b^{ij}, \beta_b^{ij}) \mu_{2}^{\sigma^{ij}_c}(\alpha_c^{ij}, \beta_c^{ij}),
\end{aligned}
\end{equation}
and check that 
\begin{equation}
\begin{aligned}\label{eq:secondCheck}
(\text{second check:}) \quad  |\mu_{abc}^{ij}| \neq 1.
\end{aligned}
\end{equation}
Assuming success on the latter, we then determine the pair $(\beta_d^{ij}, \sigma_d^{ij}) \in (0,\pi) \times \mathcal{MV}(\alpha_d^{ij},  2 \pi - \alpha^{ij}_a - \alpha^{ij}_b - \alpha^{ij}_c)$ uniquely and explicitly by the design formulas 
\begin{equation}
\begin{aligned}\label{eq:marchingFinal}
\sigma^{ij}_d1= -\sign\big( (\mu^{ij}_{abc})^2 -1  \big), \quad \beta^{ij}_d = \arccos\Big(\sigma^{ij}_d \frac{2 \mu^{ij}_{abc} - ((\mu^{ij}_{abc})^2 + 1) \cos(\alpha^{ij}_a + \alpha^{ij}_b + \alpha^{ij}_c) }{2\mu^{ij}_{abc} \cos(\alpha^{ij}_a + \alpha^{ij}_b + \alpha^{ij}_c) - ((\mu^{ij}_{abc})^2 + 1)} \Big).
\end{aligned}
\end{equation}
Finally, we turn our attention to the side lengths  $\ell_{bd}^{ij}$ and $\ell_{cd}^{ij}$.  Since all the interior angles of the $(i,j)$-panel are known,  these angles, combined with the left and bottom side lengths in (\ref{eq:marchingSides}), determine $\ell_{bd}^{ij}$ and $\ell_{cd}^{ij}$.  Precisely,
\begin{equation}
\begin{aligned}\label{eq:getSides}
\left(\begin{array}{c} \ell_{cd}^{ij} \\ \ell_{bd}^{ij} \end{array}\right) = \mathbf{L}(\alpha_{a}^{ij}, \alpha_b^{ij}, \alpha_c^{ij}) \left(\begin{array}{c} \ell_{ab}^{ij} \\ \ell_{ac}^{ij} \end{array}\right), \quad  \mathbf{L}(\alpha_{a}, \alpha_b, \alpha_c) := \left(\begin{array}{cc}  \tfrac{-\sin \alpha_b}{\sin(\alpha_a + \alpha_b + \alpha_c)}   &  \tfrac{\sin(\alpha_a + \alpha_b)}{\sin(\alpha_a + \alpha_b + \alpha_c)} \\ 
\tfrac{\sin(\alpha_a + \alpha_c)}{\sin(\alpha_a + \alpha_b + \alpha_c)} & \tfrac{-\sin \alpha_c}{\sin(\alpha_a + \alpha_b + \alpha_c)} \end{array}\right).
\end{aligned}
\end{equation}
As these are ``side lengths", this calculation is of course subject to the constraints
\begin{equation}
\begin{aligned}\label{eq:finalCheck}
(\text{final check:}) \quad \ell_{cd}^{ij} > 0 ,\quad \ell_{bd}^{ij} > 0.
\end{aligned}
\end{equation}
In total, the collective calculations in (\ref{eq:firstCheck})-(\ref{eq:finalCheck}) reveal an explicit procedure to obtain the unique crease pattern at the $(i,j)$-panel that admits RFFQM for the pre-determined sector angles (\ref{eq:marchingSector}), mountain-valley assignments (\ref{eq:marchingMV}) and side lengths (\ref{eq:marchingSides}).

Clearly then, with the outline above, we are building the core structure of a marching algorithm to obtain an overall $N \times M$ RFFQM crease pattern. There are, however, some finer points to address:   The first point is the initial input.  Notice that, for $i$ or $j =1$, some of quantities in  (\ref{eq:marchingSector}), (\ref{eq:marchingMV})  and (\ref{eq:marchingSides}) cannot be prescribed by their neighboring panels (as there is no such corresponding neighboring panel).   Instead, we treat  these parameters, i.e., 
\begin{equation}
\begin{aligned}\label{eq:inputMarching}
(\text{input data:}) \quad \begin{cases}
(\alpha_{a}^{11}, \beta_a^{11}), \ldots, (\alpha_a^{(N+1)1}, \beta_a^{(N+1)1}) \in (0,\pi)^2  \text{ and } \neq (\pi/2,\pi/2),\\
(\alpha_{a}^{12}, \beta_a^{12}), \ldots, (\alpha_a^{1(M+1)}, \beta_a^{1(M+1)}) \in (0,\pi)^2 \text{ and } \neq (\pi/2,\pi/2), \\
\sigma_a^{11} \in \mathcal{MV}(\alpha_{a}^{11}, \beta_a^{11}), \ldots, \sigma_{a}^{(N+1)1} \in \mathcal{MV}(\alpha_a^{(N+1)1}, \beta_a^{(N+1)1}), \\
\sigma_a^{12} \in \mathcal{MV}(\alpha_{a}^{12}, \beta_a^{12}), \ldots, \sigma_{a}^{1(M+1)} \in \mathcal{MV}(\alpha_a^{1(M+1)}, \beta_a^{1(M+1)}), \\
\ell_{ab}^{11}, \ldots, \ell_{ab}^{N1} >0, \\
\ell_{ac}^{11}, \ldots, \ell_{ac}^{1M} > 0,
\end{cases}
\end{aligned}
\end{equation}
 as input data to the marching algorithm.   Then, by applying all consistency conditions with these parameters (e.g., $\alpha_b^{i1} = \alpha_a^{(i+1)1}, \ldots)$, we can initialize the procedure in (\ref{eq:marchingSector})-(\ref{eq:finalCheck}) at the $(1,1)$-panel and march without any inconsistency.  The general physical picture is sketched in Fig.\;\ref{fig:Marching}(c).  The input data in (\ref{eq:inputMarching}) serves to parameterize the left and bottom boundaries of the crease pattern.  By marching, we simply discover the overall RFFQM crease pattern consistent with this data (Fig.\;\ref{fig:Marching}(d)).

To the last point, we can now properly interpret the checks (\ref{eq:firstCheck}), (\ref{eq:secondCheck}) and (\ref{eq:finalCheck}).  Given the input parameters (\ref{eq:inputMarching}) on the left and bottom boundaries and following through on the reasoning in the design theorem and its application to marching, there are only two possibilities:
\begin{itemize}
\item There is a unique $N\times M$ RFFQM crease pattern consistent with the prescribed geometry and mountain-valley assignments on the left and bottom boundary.\
\item Or there is no such crease pattern consistent ``\ldots" boundary.
\end{itemize}
The former statement is simply the result of marching, via (\ref{eq:marchingSector})-(\ref{eq:finalCheck}), without failure for the input.  Alternatively, the latter is deduced when one of the checks in  (\ref{eq:firstCheck}), (\ref{eq:secondCheck}) and (\ref{eq:finalCheck})  is violated during the process of marching.   In this case, the algorithm cannot continue, and we are left with the fact that the chosen input data is incompatible with RFFQM. 

\bigskip

\textbf{A linear program for the side lengths.}  The aforementioned incompatibility is somewhat unsatisfying; really, it calls into question of robustness of the algorithm when dealing with  crease patterns that have a large number of panels.  There are, however, some convenient tools to address incompatibility and robustness. Along this line, we build on an idea of Lang and Howell \cite{lang2018rigidly}. The basic heuristic is that this incompatibility is unfortunately extremely sensitive to the side length conditions (\ref{eq:finalCheck}).   To this point, notice that the calculation of sector angle and mountain-valley assignments  (\ref{eq:firstCheck})-(\ref{eq:marchingFinal}) does not require the lengths.  We may therefore compute by marching all the sector angles $\boldsymbol{\alpha}_{Tot} := (\alpha_{a}^{11}, \alpha_b^{11}, \ldots, \alpha_b^{NM}, \alpha_c^{NM})$ and all the mountain-valley assignments $\boldsymbol{\sigma} := (\sigma_a^{11}, \ldots, \sigma_a^{NM})$ prior to dealing with the lengths.   Then, rather than treat the side lengths  $\ell_{ab}^{11}, \ldots, \ell_{ab}^{N1} >0$ and 
$\ell_{ac}^{11}, \ldots, \ell_{ac}^{1M} > 0$ as input, we can attempt to optimize these subject to the sensitive criterion in (\ref{eq:finalCheck}) that all the lengths of the overall crease pattern are positive.

This optimization can conveniently be formulated as a standard linear program.  In this direction, we introduce the vector $\boldsymbol{\ell} := ( \ell_{ab}^{11}, \ldots, \ell_{ab}^{N1}, \ell_{ac}^{11}, \ldots, \ell_{ac}^{1M})$.  It then follows from (\ref{eq:getSides}) that there exists a linear transformation $\mathbf{L}_{\boldsymbol{\alpha}_{Tot}}^{ij}$, depending only on the sector angles $\boldsymbol{\alpha}_{Tot}$,  such that 
\begin{equation}
( \ell_{cd}^{ij}, ~ \ell_{bd}^{ij} ) = \mathbf{L}_{\boldsymbol{\alpha}_{Tot}}^{ij} \boldsymbol{\ell}.
\end{equation}
By also introducing the vectors $\mathbf{1} = (1, \ldots, 1)$ and $\boldsymbol{\ell}_{\varepsilon}$ with the same dimension as $\boldsymbol{\ell}$, we obtain the standard linear program\footnote{In keeping with the literature on linear programming, the inequalities here represent an element-wise constraint.} in $\boldsymbol{\ell}_{\varepsilon}$,
\begin{equation}
\begin{aligned}
\begin{cases}\label{eq:linearP}
\min  \boldsymbol{\ell}_{\varepsilon}\mathbf{1}^T  \\
 \text{subject to } \mathbf{L}_{\boldsymbol{\alpha}_{Tot}}^{ij}\boldsymbol{\ell}_{\varepsilon}  \geq  \varepsilon (1,1) \text{ for }  i \in \{ 1, \ldots, N\}, j \in \{1, \ldots, M\}, \\
 \text{and } \boldsymbol{\ell}_{\varepsilon} \geq 0 
 \end{cases}
\end{aligned}
\end{equation}
for $\varepsilon \geq 0$.
 As is well-known, standard and computationally efficient algorithms can quickly evaluate the feasibility of the domain of the inequalities above and, if feasible, compute a solution to the optimization.  In short, this linear program provides a robust and efficient framework for ensuring a valid crease pattern or declaring that such a pattern is infeasible.  Particularly, in assuming a case $\varepsilon > 0$ is feasible and has solution $\boldsymbol{\ell}^{\star}_{\varepsilon}$ to (\ref{eq:linearP}), we can set  
\begin{equation}
\begin{aligned}
 ( \ell_{ab}^{11}, \ldots, \ell_{ab}^{N1}, \ell_{ac}^{11}, \ldots, \ell_{ac}^{1M})  = \boldsymbol{\ell}^{\star}_{\varepsilon} + \delta \mathbf{1}.
\end{aligned}
\end{equation}
By continuity and using that $\boldsymbol{\ell}^{\star}_{\varepsilon}$ obeys the constraints in (\ref{eq:linearP}), there is a $\delta_{\varepsilon} > 0$ such that  for all $\delta \in (0, \delta_{\varepsilon})$, the side lengths $(\ell_{cd}^{ij}, \ell_{bd}^{ij}) := \mathbf{L}_{\boldsymbol{\alpha}_{Tot}}^{ij}(\boldsymbol{\ell}^{\star}_{\varepsilon} + \delta \mathbf{1})$ are strictly positive for all $i,j$.  So by choosing $\delta$ accordingly, we obtain a valid crease pattern by this procedure for a given set of sector angles $\boldsymbol{\alpha}_{Tot}$.  On the other hand, if this linear program fails the  feasibility test for the choice $\varepsilon = 0$, then we can be assured that the sector angles $\boldsymbol{\alpha}_{Tot}$, obtained by marching, cannot produce a valid crease pattern.

In practice, while capable of quickly assessing the feasibility of inputed sector angles $\boldsymbol{\alpha}_{Tot}$, this program is no substitute for sound engineering judgment or an ``Origamist" intuition.  In fact, from our exploration of the configuration space by this method, it appears that most feasible sector angles will result in a crease pattern with panels of  disparate aspect ratios. To illustrate this point, consider the example in Fig.\;\ref{fig:Random}.  Here, we input boundary sector angles in (\ref{eq:inputMarching}) by a uniform distribution whose average gives the classical Miura-Ori and  whose support is $\pm0.05 \pi$ to the left and right of this average.  We also input the mountain-valley assignments consistent with a Miura-Ori.  We then march, as above, to obtain all the sector angles $\boldsymbol{\alpha}_{Tot}$ and mountain-valley assignments $\boldsymbol{\sigma}$, and we use the sector angles to compute the side lengths via the optimization (\ref{eq:linearP}).   One such result of this procedure\textemdash and a seemingly generic one at that\textemdash is the crease pattern in Fig.\;\ref{fig:Random}.  Notice that some of the side lengths emerging from the optimization are driven close to zero, while others take on rather large values.  This is hardly in keeping with manufacturability considerations. 

A natural next step might be to replace the (relatively arbitrary) linear objective function in (\ref{eq:linearP}) with a non-linear one that attempts to keep the aspect ratios of the panels as uniform as possible, while optimizing over the same feasible set.  However, this generalization would be a significant departure from the well-developed and efficient algorithms used for the linear program.  So we leave it a topic of future research.  In any case, one can always start from a well-known crease pattern (e.g., the Miura-Ori), and build a new pattern by slight perturbations that do not necessarily require implementation of the linear program for robustness.  In fact, we follow this approach\textemdash with marked success\textemdash when discussing the inverse problem in the next section.

\begin{figure}
\centering
\includegraphics[width = 6.5in]{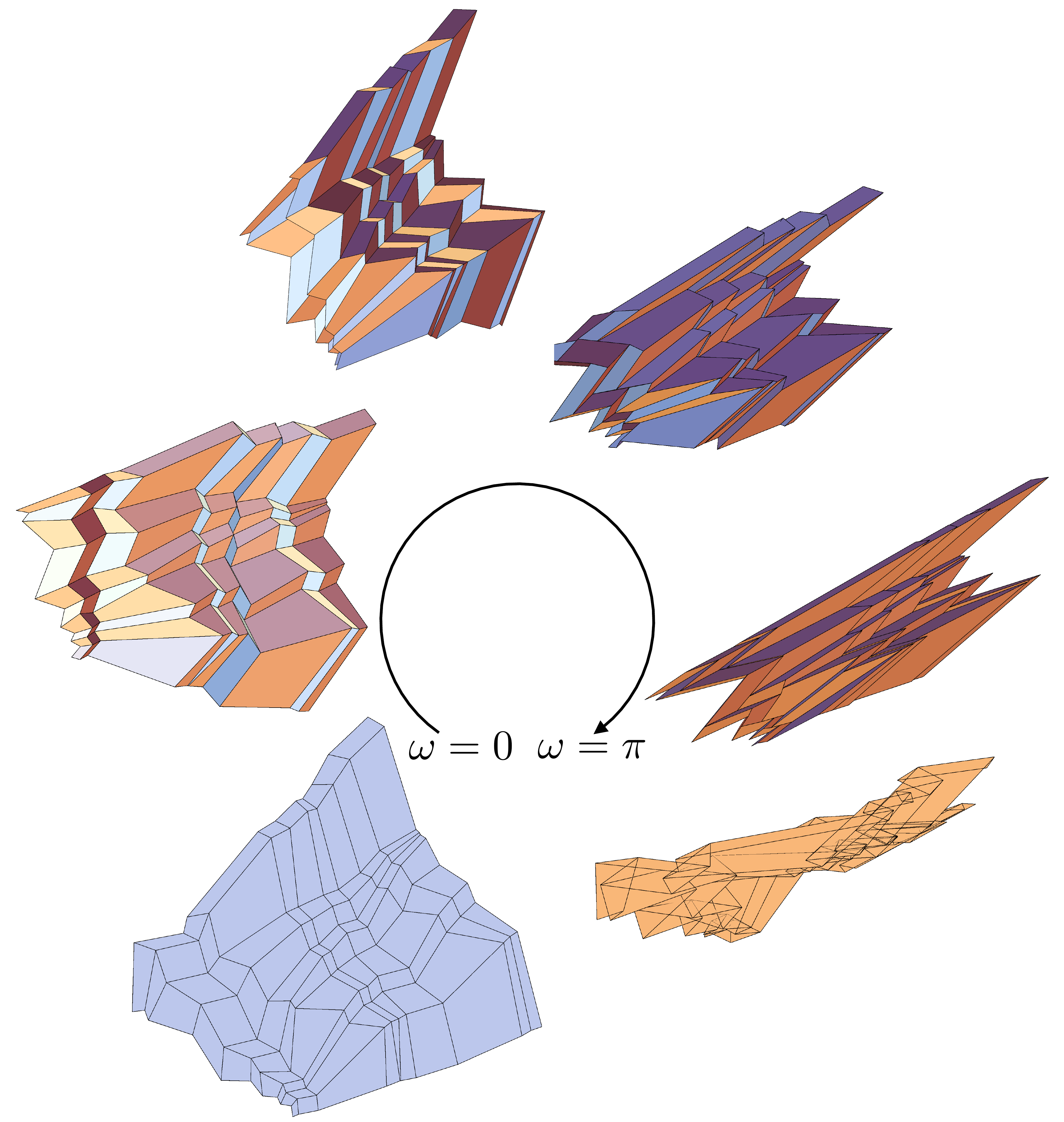}
\caption{A ``generic" RFFQM crease pattern obtained the marching algorithm.  The input sector angles are chosen from a uniform distribution that perturbs around a classical Miura-Ori,  the side lengths are determined by optimization using the linear program, and the mountain-valley assignment is consistent with the Miura-Ori.}
\label{fig:Random}
\end{figure}

\bigskip

\textbf{Marching to obtain the kinematics.}  Let us assume we have completed the marching procedure above\textemdash either by incorporating the linear program or not\textemdash and obtained all sectors angels, mountain-valley assignments and side lengths at every $(i,j)$-panel in the $N \times M$ crease pattern.   It follows that every vertex of the crease pattern is determined up to rigid motion.  For an explicit calculation, we introduce the rotation $\mathbf{R}_{\mathbf{e}_3}(\alpha) := \mathbf{e}_3 \otimes \mathbf{e}_3 + \cos \alpha (\mathbf{I} - \mathbf{e}_3 \otimes \mathbf{e}_3) + \sin \alpha (\mathbf{e}_2 \otimes \mathbf{e}_1 - \mathbf{e}_1 \otimes \mathbf{e}_2)$  and prescribe the $(1,1)$-panel by the vertices on $\mathbb{R}^3$,  
\begin{equation}
\begin{aligned}\label{eq:creaseInputs}
\mathbf{x}_a^{11} = \mathbf{0}, \quad \mathbf{x}_c^{11} = \ell_{ac}^{11} \mathbf{e}_1, \quad \mathbf{x}_b^{11} =\tfrac{\ell_{ab}^{11}}{\ell_{ac}^{11}} \mathbf{R}_{\mathbf{e}_3}(\alpha_{a}^{11}) (\mathbf{x}_c^{11} - \mathbf{x}_a^{11}), \quad \mathbf{x}_d^{11} = \mathbf{x}_b +  \tfrac{\ell_{bd}^{11}}{\ell_{ab}^{11}}  \mathbf{R}_{\mathbf{e}_3}(\alpha_{b}^{11}) (\mathbf{x}_a^{11} - \mathbf{x}_b^{11})
\end{aligned}
\end{equation}
(see Fig.\;\ref{fig:Marching}(a-b) for the vertex notation).
This then fixes the rigid body motion and enables a marching procedure for all the other vertices.  In particular, the remaining panels along the first row are described by the vertices
\begin{equation}
\begin{aligned}
&\mathbf{x}_a^{1j} = \mathbf{x}_c^{1(j-1)}, \quad \mathbf{x}_b^{1j} = \mathbf{x}_d^{1(j-1)}, \quad \mathbf{x}_c^{1j} = \mathbf{x}_a^{1j} + \tfrac{\ell_{ac}^{1j}}{\ell_{ab}^{1j}} \mathbf{R}_{\mathbf{e}_3}( -\alpha^{1j}_a)( \mathbf{x}_b^{1j} - \mathbf{x}_a^{1j}) , \\
&\qquad  \mathbf{x}_d^{1j} = \mathbf{x}_b^{1j}  + \tfrac{\ell_{bd}^{1j}}{\ell_{ab}^{1j}} \mathbf{R}_{\mathbf{e}_3}(\alpha_b^{1j}) (\mathbf{x}_a^{1j} - \mathbf{x}_b^{1j}),
\end{aligned}
\end{equation}
and the panels for the rest of the rows are described by the vertices
\begin{equation}
\begin{aligned}
&\mathbf{x}_a^{ij} = \mathbf{x}_b^{(i-1)j}, \quad \mathbf{x}_c^{ij} = \mathbf{x}_d^{(i-1)j}, \quad \mathbf{x}_b^{ij} = \mathbf{x}^{ij}_a +  \tfrac{\ell^{ij}_{ab}}{\ell^{ij}_{ac}} \mathbf{R}_{\mathbf{e}_3}(\alpha_a^{ij}) (\mathbf{x}_c^{ij} - \mathbf{x}_a^{ij}),\\
&\qquad  \mathbf{x}_d^{ij} = \mathbf{x}_b^{ij} + \tfrac{\ell^{ij}_{bd}}{\ell^{ij}_{ab}} \mathbf{R}_{\mathbf{e}_3}(\alpha_b^{ij})(\mathbf{x}_a^{ij} - \mathbf{x}_b^{ij}), \quad i  >1.
\end{aligned}
\end{equation}

From the vertices, we can build the full kinematics of the RFFQM by making repeated use of the local description of the origami in Theorem \ref{VertexTheorem} and following the schematic in Fig.\;\ref{fig:Marching}(b).  Basically at every $\mathbf{x}_a^{ij}$ vertex, we can define the local tangents as
\begin{equation}
\begin{aligned}\label{eq:tansLocal}
\mathbf{t}_1 := \frac{\mathbf{x}_c^{ij}  -\mathbf{x}_a^{ij}}{|\mathbf{x}_b^{ij}  -\mathbf{x}_a^{ij}|}, \quad \mathbf{t}_2 := \frac{\mathbf{x}_b^{ij}  -\mathbf{x}_a^{ij}}{|\mathbf{x}_b^{ij}  -\mathbf{x}_a^{ij}|}, \quad \mathbf{t}_3 := \mathbf{R}_{\mathbf{e}_3}(\beta_a^{ij}) \mathbf{t}_2, \quad \mathbf{t}_4 :=  \mathbf{R}_{\mathbf{e}_3}(\beta_a^{ij}- \pi) \mathbf{t}_1.
\end{aligned}
\end{equation}  
(Here, we do not label them with an $ij$ since there is never a need to store these quantities.)  We can then consider the four folding angles along the creases associated to these tangent, i.e., $\gamma_1^{ij}$, $\gamma_2^{ij}$, $\gamma_3^{ij}$ and $\gamma_4^{ij}$, and observe that these angles have to be consistent with the local  description in the theorem and with each other when marching from neighbor to neighbor.  On the last point, notice that $\gamma_{3}^{ij}$ must coincide with $\gamma_{1}^{i(j-1)}$ since they both describe the same crease. Likewise $\gamma_4^{ij} = \gamma_2^{(i-1)j}$.  This leads to an explicit procedure.  We define 
\begin{equation}
\begin{aligned}\label{eq:foldInitial}
\gamma_{1\omega}^{11}  =\omega, \quad \gamma_{2\omega}^{11} = \bar{\gamma}_2^{\sigma_a^{11}}(\gamma_{1\omega}^{11}; \alpha_a^{11}, \beta_a^{11}), \quad \gamma_{3\omega}^{11} = -\sigma_a^{11} \gamma_{1\omega}^{11}, \quad \gamma_{4\omega}^{11} = \sigma_a^{11} \gamma_{2\omega}^{11}
\end{aligned}
\end{equation}
at the $\mathbf{x}_a^{11}$ vertex to initialize a procedure for these kinematics for some $\omega \in [-\pi, \pi]$.   The remaining folding angles at the bottom boundary are then determined by
\begin{equation}
\begin{aligned}\label{eq:foldMiddle}
\gamma_{3\omega}^{1j} = \gamma_{1\omega}^{1(j-1)}, \quad \gamma_{1\omega}^{1j} = -\sigma_a^{1j} \gamma_{3\omega}^{1j}, \quad  \gamma_{2\omega}^{1j} = \bar{\gamma}_2^{\sigma_a^{1j}}(\gamma_{1\omega}^{1j}; \alpha_a^{1j}, \beta_a^{1j}), \quad  \gamma_{4\omega}^{1j} = \sigma_a^{1j} \gamma_{2\omega}^{1j},
\end{aligned}
\end{equation}
and the rest of the folding angles needed to build the kinematics  are  completely determined by 
\begin{equation}
\begin{aligned}\label{eq:foldFinal}
\gamma_{4\omega}^{ij} = \gamma_{2\omega}^{(i-1)j}, \quad \gamma_{2\omega}^{ij} = \sigma_a^{ij} \gamma_{4\omega}^{ij}, \quad \gamma_{1\omega}^{ij} = \bar{\gamma}_1^{-\sigma_a^{ij}}( \gamma_{2\omega}^{ij}; \alpha_a^{ij}, \beta_a^{ij}), \quad \gamma_{3\omega}^{ij} = -\sigma_a^{ij} \gamma_{1\omega}^{ij}, \quad i > 1.
\end{aligned}
\end{equation}

From the explicit description of the folding angles, we can compute the deformation gradients by a marching procedure.  To this end, let us focus on the sketch in Fig.\;\ref{fig:Marching}(b) for an $(i,j)$-panel that is not the $(1,1)$-panel.  We have access to either the deformation gradient $\mathbf{R}_\omega^{i(j-1)}$ of the left neighbor or  the deformation gradient $\mathbf{R}_\omega^{(i-1)j}$ of the bottom neighbor (or both), as depicted.  We also know that the deformation gradients must, up to an overall rigid rotation, be as indicated in Fig.\;\ref{fig:SingleVertex}(b) for this vertices folding angles. Given the local tangents in (\ref{eq:tansLocal}), the rotations in (\ref{eq:localRots}), and  folding angles (\ref{eq:foldInitial})-(\ref{eq:foldFinal}), we therefore work out that 
\begin{equation}
\begin{aligned}\label{eq:defGrads}
\mathbf{R}_{\omega}^{ij} = \begin{cases}
\mathbf{R}_{\omega}^{i(j-1)} \mathbf{R}_2(-\gamma_{2\omega}^{ij}) & \text{ if $\mathbf{R}_{\omega}^{i(j-1)}$  is known,} \\
\mathbf{R}_{\omega}^{(i-1)j} \mathbf{R}_1(\gamma_{1\omega}^{ij}) & \text{ if $\mathbf{R}_{\omega}^{(i-1)j}$  is known.}
\end{cases}
\end{aligned}
\end{equation}
Note, when both deformation gradients are known, i.e., $i,j>1$, the two prescriptions for $\mathbf{R}_{\omega}^{ij}$ above will coincide since the sector angles and mountain valley assignments have been chosen to satisfy panel compatibility.  We therefore set $\mathbf{R}_\omega^{11} = \mathbf{I}$ to eliminate the degeneracy of a rigid rotation and use this to initialize the formulas in (\ref{eq:defGrads}) to compute the deformation gradients at all panels.

With the deformation gradients, we are finally able to compute the deformations.  Explicitly, we take the first panel to be undeformed, i.e., 
\begin{equation}
\begin{aligned}
\mathbf{y}_{a\omega}^{11} = \mathbf{x}_{a}^{11}, \quad \mathbf{y}_{b\omega}^{11} = \mathbf{x}_b^{11}, \quad \mathbf{y}_{c\omega}^{11} = \mathbf{x}_c^{11} , \quad \mathbf{y}_{d\omega}^{11} = \mathbf{x}_d^{11},
\end{aligned}
\end{equation}
as $\mathbf{R}_\omega^{11} = \mathbf{I}$.  
Then, every other panel along the first row has deformed vertices 
\begin{equation}
\begin{aligned}
&\mathbf{y}_{a\omega}^{1j} = \mathbf{y}_{c\omega}^{1(j-1)}, \quad \mathbf{y}_{b\omega}^{1j} = \mathbf{y}_{a\omega}^{1j} + \mathbf{R}_\omega^{1j}(\mathbf{x}^{1j}_b - \mathbf{x}^{1j}_a), \quad \mathbf{y}_{c\omega}^{1j} = \mathbf{y}_{a\omega}^{1j} + \mathbf{R}_\omega^{1j}(\mathbf{x}^{1j}_c - \mathbf{x}^{1j}_a),  \\
&\qquad  \mathbf{y}_{d\omega}^{1j} = \mathbf{y}_{a\omega}^{1j} + \mathbf{R}_\omega^{1j}(\mathbf{x}^{1j}_d - \mathbf{x}^{1j}_a).   
\end{aligned}
\end{equation}
Finally, every remaining panel has deformed vertices 
\begin{equation}
\begin{aligned}\label{eq:finalCalc}
&\mathbf{y}_{a\omega}^{ij} = \mathbf{y}_{b\omega}^{(i-1)j}, \quad \mathbf{y}_{b\omega}^{ij} = \mathbf{y}_{a\omega}^{ij} + \mathbf{R}_\omega^{ij}(\mathbf{x}^{ij}_b - \mathbf{x}^{ij}_a), \quad \mathbf{y}_{c\omega}^{ij} = \mathbf{y}_{a\omega}^{ij} + \mathbf{R}_\omega^{ij}(\mathbf{x}^{ij}_c - \mathbf{x}^{ij}_a),  \\
&\qquad  \mathbf{y}_{d\omega}^{ij} = \mathbf{y}_{a\omega}^{ij} + \mathbf{R}_\omega^{ij}(\mathbf{x}^{ij}_d - \mathbf{x}^{ij}_a).
\end{aligned}
\end{equation}

This completes the marching procedure to obtain the kinematics for a given $\omega \in [-\pi,\pi]$.  By repeating the procedure monotonically from $\omega= 0$ to $\omega = \pi$, we obtain the complete kinematical description of the rigidly and flat-foldable motion of the crease pattern (e.g., Fig.\;\ref{fig:Marching}(e)).

\subsection{Discussion on the inverse problem}\label{ssec:Inverse}

Beginning with the pioneering work in \cite{klein2007shaping} on non-Euclidean plates, the design of shape-morphing structures and materials has blossomed into an active topic of research.  In this context, the forward problem refers to computing the shapes achievable given an explicit design of the material or structure.  This is essentially what we have outlined thus far in the context of RFFQM.  Alternatively, the inverse problem refers to designing a material or structure to achieve a targeted shape. This is a far more difficult task. Nevertheless, some results are beginning to emerge in this direction for stimuli-responsive materials \cite{aharoni2018universal, griniasty2019curved, plucinsky2016programming, van2017growth}, kirigami \cite{choi2019programming}, and origami \cite{demaine2017origamizer, dudte2016programming, lang1996computational}.


We focus here on a brief introduction of the inverse problem for RFFQM, with more comprehensive exposition to follow \cite{dang_inverse}.  In contrast to these previous works, the shape-morphing achieved by RFFQM has two novel properties:
\begin{itemize}
\item For a computed design, the full (ideal) kinematics along a prescribed mountain-valley assignments are known; apply the procedure outlined in the previous section.
\item The structure is deployable; that is, it can deform from flat to folded flat as rigidly foldable origami.
\end{itemize}
These aspects are expected to be important in many aforementioned applications \cite{arya2016packaging,arya2016ultralight, kim2018printing,kuribayashi2006self,miura1985method, pellegrino2014deployable}, as they enable predictable deployment from both an the easy-to-manufacture flat state and a fully-folded compact state.  

\begin{figure}
\centering
\includegraphics[width = 6in]{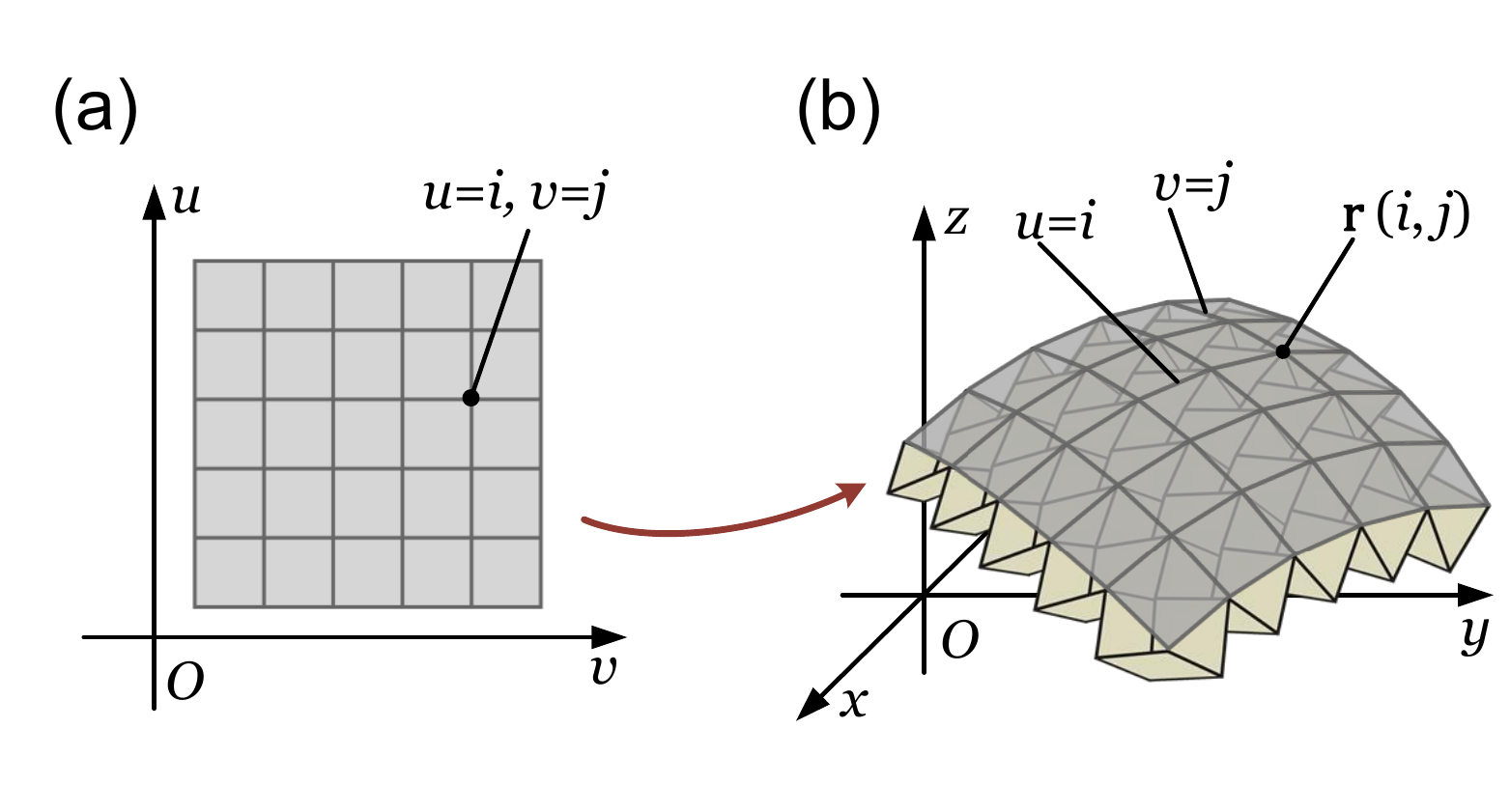}
\caption{Sketch of the parameterizations in the inverse design problem. (a) The targeted surface is discretized by a parameterization on rectangular subset of the $\mathbb{Z}^2$ lattice. (b) Likewise the origami is discretized on this lattice, but with the corner vertices of $2 \times 2$ sets of panels forming the image points of the lattice.}
\label{fig:finite_diff}
\end{figure}

We start simple to build some understanding of feasibility in the inverse problem. Recall that the configuration space for RFFQM depends on the sector angles, side lengths and mountain-valley assignments at the left and bottom boundary (equation (\ref{eq:inputMarching}) and Fig.\;\ref{fig:Marching}(c)). Recall also that the algorithm to explore this configuration space may not admit a valid crease pattern, or, even when it does admit one, there is an entire one-parameter family of deformations associated to this crease pattern.  In other words, the configuration space is simultaneously massive and fraught with robustness issues.  To alleviate these concerns, we restrict our focus here  to patterns consisting of two properties:
\begin{enumerate}
\item[(i)] We study slight perturbations of the boundary sector angles and lengths of a basic Miura-Ori. 
\item[(ii)] We fix the boundary mountain-valley assignments to be consistent with a basic Miura-Ori.
\end{enumerate}
These assumptions are built on the heuristic that, while a crease pattern can look qualitatively similar to a basic Miura-Ori in the flat state, it can be dramatically different by comparison in its partially folded states due to the highly non-linear nature of origami kinematics.  We show that this ``dramatic difference" can be harnessed to achieve a targeted shape.

\begin{figure}[t!]
\centering
\includegraphics[width =\textwidth]{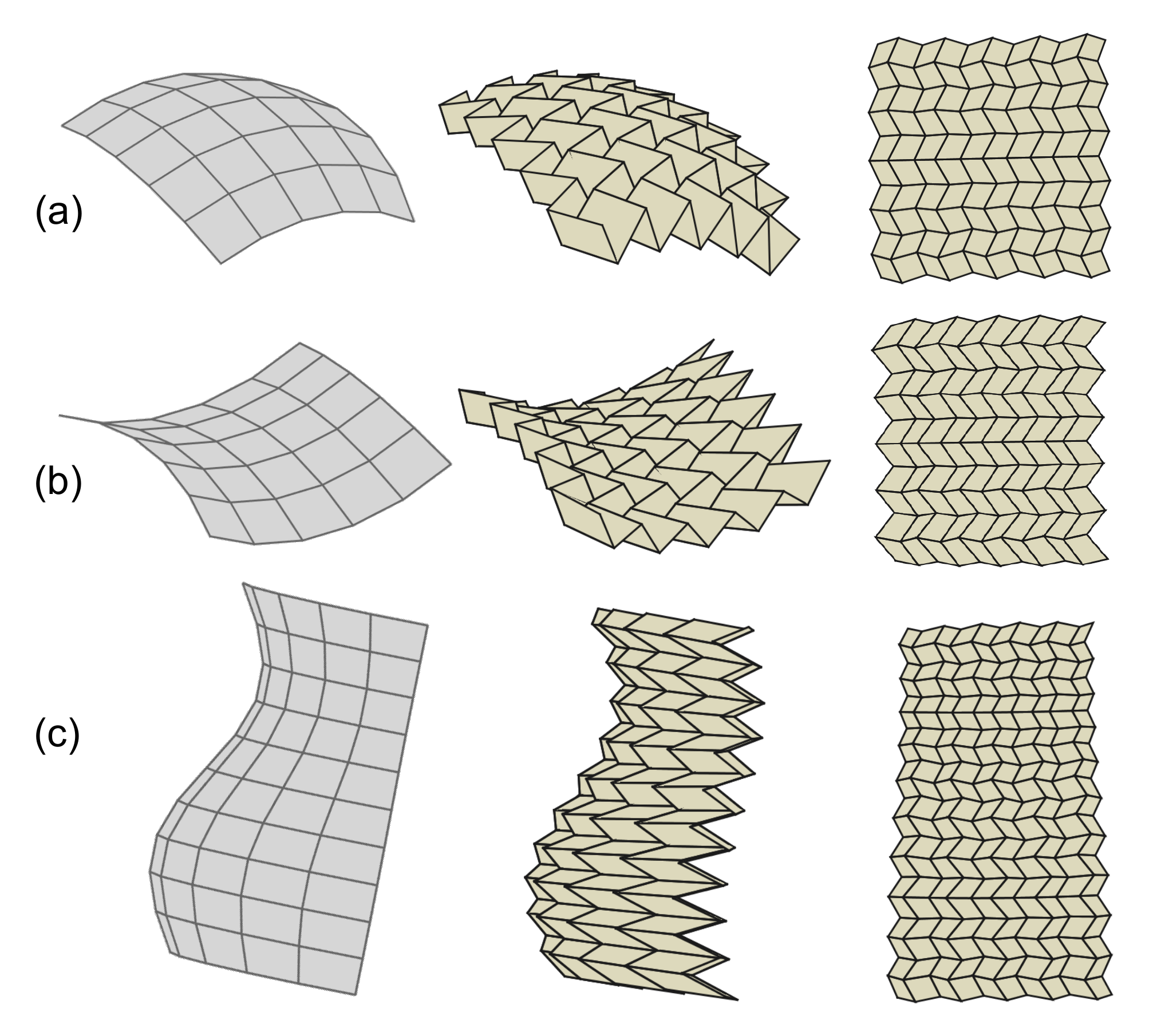}
\caption{Examples of RFFQM to achieve a Spherical Cap (a), a Hyperboloid (b) and a Quarter Vase (c). (Left:) The smooth targeted surface and its discrete mesh. (Middle:) A RFFQM in it optimized partially folded state. (Right:) The flat crease pattern that is slightly perturbed from a Miura-Ori pattern and rigidly folded to achieve the middle state.}
\label{fig:InverseExamples}
\end{figure}

Our basic formulation for the inverse design problem follows the sketch in Fig.\;\ref{fig:finite_diff}. Consider a targeted surface as the image of a sufficiently smooth map $\hat{\mathbf{r}} \colon (0,N) \times (0,M) \rightarrow \mathbb{R}^3$, as depicted in grey.  At every interior pair of integers, we compute, by a finite difference scheme, discrete versions of quantities describing the shape of target surface (e.g., the first and second fundamental forms), which we collect in a vector denoted $\hat{\mathbf{S}}_{i,j}$.  We also study $2N \times 2M$ origami tessellations obtained by the marching procedure above, and we index the vertices of these partially folded origami as $\mathbf{y}_{\boldsymbol{\alpha}, \boldsymbol{\beta}, \boldsymbol{\ell}, \omega}^{ij}$ for $\boldsymbol{\alpha}, \boldsymbol{\beta}, \boldsymbol{\ell}$ parameterizing the relevant boundary data in\footnote{replacing $N$ by $2N$ and $M$ by $2M$ for consistency, of course} (\ref{eq:inputMarching}) and $\omega$ parameterizing the folding parameter.  We then introduce the reparameterization of the origami, $\mathbf{r}_{\boldsymbol{\alpha}, \boldsymbol{\beta}, \boldsymbol{\ell}, \omega} \colon \{ 0,1,\ldots, N\} \times \{ 0,1, \ldots, M\} \rightarrow \mathbb{R}^3$, such that 
\begin{equation}
\begin{aligned}
\mathbf{r}_{\boldsymbol{\alpha}, \boldsymbol{\beta}, \boldsymbol{\ell}, \omega} (i,j) := \mathbf{y}_{\boldsymbol{\alpha}, \boldsymbol{\beta}, \boldsymbol{\ell}, \omega}^{(2i)(2j)}.
\end{aligned}
\end{equation}
Importantly, this new parameterization serves to smooth out some of the  oscillatory behavior of the origami by taking $2 \times 2$ sets of panels as the fundamental meshing. As a result, we have a reasonable comparison between the vector $\hat{\mathbf{S}}_{i,j}$ of our targeted surface and the analogous vector $\mathbf{S}_{i,j}(\boldsymbol{\alpha}, \boldsymbol{\beta}, \boldsymbol{\ell}, \omega)$ obtained by computing the same finite difference quantities for $\mathbf{r}_{\boldsymbol{\alpha}, \boldsymbol{\beta}, \boldsymbol{\ell}, \omega}$.   In this comparison, we treat our objective function as
\begin{equation}
\begin{aligned}\label{eq:optInverse}
Loss( \boldsymbol{\alpha}, \boldsymbol{\beta}, \boldsymbol{\ell}, \omega) \sim  \sum_{i = 0, \ldots, M} \sum_{j = 0, \ldots, N} |\mathbf{S}_{i,j}(\boldsymbol{\alpha}, \boldsymbol{\beta}, \boldsymbol{\ell}, \omega) - \hat{\mathbf{S}}_{i,j}|^2,
\end{aligned}
\end{equation}
up to some numerical considerations not developed here (see \cite{dang_inverse}).  Importantly, $Loss(\boldsymbol{\alpha}, \boldsymbol{\beta}, \boldsymbol{\ell}, \omega) = 0$ provides two discretizations $\hat{\mathbf{r}}$ and $\mathbf{r}_{\boldsymbol{\alpha}, \boldsymbol{\beta}, \boldsymbol{\ell}, \omega}$ that are, essentially, the same up to rigid body motion. As a consequence,  the global minimizers to (\ref{eq:optInverse})  represent a reasonable notion of \textit{solutions to the inverse problem}.  While global minimality is desirable, we aim for the more modest goal of computing local minimizers to (\ref{eq:optInverse});  particularly, for the small subset RFFQM crease patterns that are slight perturbations of the basic Miura-Ori as stated in (i-ii).  This setting allows us to investigate the inverse problem by means of  standard optimization tools.

It is well-known that smooth isometric embeddings are extremely rigid and incapable of adeptly approximating surfaces with non-zero Gaussian curvature (Gauss's {\it Theorema Egregium}). However, this rigidity breaks down when the isometric condition or smoothness is relaxed, as Nash's famed embedding theorem \cite{nash1954c1} shows. To achieve non-trivial Gaussian curvature, one can exploit non-isometric deformation induced by active materials \cite{mostajeran2016encoding, kowalski2018curvature}, or non-smooth isometric embeddings induced by origami/kirigami constructions \cite{CALLENS2018241}.
RFFQM is a version of the latter approach, and so there is no apparent mathematical obstacle to their utility as approximations of arbitrary surfaces.  Yet, approximation theorems involving origami  maps\textemdash of which the most general can be found in  \cite{conti2008confining}\textemdash often involve significant refinement, making their application ill-suited to practical engineering design.  We are therefore motivated to study whether structured flat-foldable origami, as exhibited by the assumption (i-ii), can serve as good approximations to surfaces of varied and non-zero Gaussian curvature. 

In this direction, we test the design methodology outlined above on a series of simple examples that illustrate the range of Gaussian curvatures achievable: a Spherical Cap to achieve positive Gaussian curvature, a Hyperboloid to achieve negative Gaussian curvature, and finally a Quarter Vase to exhibit changes in Gaussian curvature.  The results of our optimization are provided in Fig.\;\ref{fig:InverseExamples}.  We find striking qualitative agreement in these examples, especially given the fact that,  with properties (i-ii), we have confined ourselves to only a small fraction of the configuration space of RFFQM.  

To this point, natural extensions to the underlying assumptions (i-ii) can be made by incorporating ideas of \textit{objective origami}, which we first explored in \cite{feng2019phase} and are developing in more detail in \cite{soorOri}.  By incorporating these extensions and tools from data-driven engineering, we are in the process of gaining a more complete picture of the configuration space.  The results given here suggest this is a promising framework for shape-morphing based on inverse design principles.

\section*{Acknowledgement}

F.F., R.D.J., and P.P. gratefully acknowledge the support of the Air Force Office of Scientific Research through the MURI grant no. FA9550-16-1-0566. The authors would like to thank Robert J. Lang for many helpful discussions and useful comments on this manuscript.

\bibliographystyle{abbrv}

\bibliography{OriBib}

\begin{thebibliography}{10}

\bibitem{aharoni2018universal}
H.~Aharoni, Y.~Xia, X.~Zhang, R.~D. Kamien, and S.~Yang.
\newblock Universal inverse design of surfaces with thin nematic elastomer
  sheets.
\newblock {\em Proceedings of the National Academy of Sciences},
  115(28):7206--7211, 2018.

\bibitem{arya2016packaging}
M.~Arya.
\newblock {\em Packaging and deployment of large planar spacecraft structures}.
\newblock PhD thesis, California Institute of Technology, 2016.

\bibitem{arya2016ultralight}
M.~Arya, N.~Lee, and S.~Pellegrino.
\newblock Ultralight structures for space solar power satellites.
\newblock In {\em 3rd AIAA Spacecraft Structures Conference}, page 1950, 2016.

\bibitem{ball1989fine}
J.~M. Ball and R.~D. James.
\newblock Fine phase mixtures as minimizers of energy.
\newblock In {\em Analysis and Continuum Mechanics}, pages 647--686. Springer,
  1989.

\bibitem{berry2019topological}
M.~Berry, M.~Lee-Trimble, and C.~Santangelo.
\newblock Topological transitions in the configuration space of non-euclidean
  origami.
\newblock {\em arXiv preprint arXiv:1910.01008}, 2019.

\bibitem{bhattacharya2003microstructure}
K.~Bhattacharya et~al.
\newblock {\em Microstructure of martensite: why it forms and how it gives rise
  to the shape-memory effect}, volume~2.
\newblock Oxford University Press, 2003.

\bibitem{CALLENS2018241}
S.~J. Callens and A.~A. Zadpoor.
\newblock From flat sheets to curved geometries: Origami and kirigami
  approaches.
\newblock {\em Materials Today}, 21(3):241 -- 264, 2018.

\bibitem{Chen396}
Y.~Chen, R.~Peng, and Z.~You.
\newblock Origami of thick panels.
\newblock {\em Science}, 349(6246):396--400, 2015.

\bibitem{choi2019programming}
G.~P. Choi, L.~H. Dudte, and L.~Mahadevan.
\newblock Programming shape using kirigami tessellations.
\newblock {\em Nature materials}, 18(9):999--1004, 2019.

\bibitem{conti2008confining}
S.~Conti and F.~Maggi.
\newblock Confining thin elastic sheets and folding paper.
\newblock {\em Archive for Rational Mechanics and Analysis}, 187(1):1--48,
  2008.

\bibitem{cui2006combinatorial}
J.~Cui, Y.~S. Chu, O.~O. Famodu, Y.~Furuya, J.~Hattrick-Simpers, R.~D. James,
  A.~Ludwig, S.~Thienhaus, M.~Wuttig, Z.~Zhang, et~al.
\newblock Combinatorial search of thermoelastic shape-memory alloys with
  extremely small hysteresis width.
\newblock {\em Nature materials}, 5(4):286--290, 2006.

\bibitem{dang_inverse}
X.~Dang, F.~Feng, P.~Plucinsky, H.~Duan, and J.~Wang.
\newblock Inverse design of rigidly and flat-foldable origami for approximating
  arbitrary surfaces.
\newblock {\em preprint}.

\bibitem{demaine2000folding}
E.~D. Demaine, M.~L. Demaine, and J.~S. Mitchell.
\newblock Folding flat silhouettes and wrapping polyhedral packages: New
  results in computational origami.
\newblock {\em Computational Geometry}, 16(1):3--21, 2000.

\bibitem{demaine2017origamizer}
E.~D. Demaine and T.~Tachi.
\newblock Origamizer: A practical algorithm for folding any polyhedron.
\newblock In {\em 33rd International Symposium on Computational Geometry (SoCG
  2017)}. Schloss Dagstuhl-Leibniz-Zentrum fuer Informatik, 2017.

\bibitem{dieleman2020jigsaw}
P.~Dieleman, N.~Vasmel, S.~Waitukaitis, and M.~van Hecke.
\newblock Jigsaw puzzle design of pluripotent origami.
\newblock {\em Nature Physics}, 16(1):63--68, 2020.

\bibitem{dudte2016programming}
L.~H. Dudte, E.~Vouga, T.~Tachi, and L.~Mahadevan.
\newblock Programming curvature using origami tessellations.
\newblock {\em Nature materials}, 15(5):583--588, 2016.

\bibitem{evans2015rigidly}
T.~Evans, R.~Lang, S.~Magleby, and L.~Howell.
\newblock Rigidly foldable origami gadgets and tessellations.
\newblock {\em Royal Society Open Science}, 2, 09 2015.

\bibitem{felton2014method}
S.~Felton, M.~Tolley, E.~Demaine, D.~Rus, and R.~Wood.
\newblock A method for building self-folding machines.
\newblock {\em Science}, 345(6197):644--646, 2014.

\bibitem{feng2019phase}
F.~Feng, P.~Plucinsky, and R.~D. James.
\newblock Helical miura origami.
\newblock {\em Phys. Rev. E}, 101:033002, 2020.

\bibitem{FILIPOV201726}
E.~Filipov, K.~Liu, T.~Tachi, M.~Schenk, and G.~Paulino.
\newblock Bar and hinge models for scalable analysis of origami.
\newblock {\em International Journal of Solids and Structures}, 124:26 -- 45,
  2017.

\bibitem{friesecke2002theorem}
G.~Friesecke, R.~D. James, and S.~M{\"u}ller.
\newblock A theorem on geometric rigidity and the derivation of nonlinear plate
  theory from three-dimensional elasticity.
\newblock {\em Communications on Pure and Applied Mathematics: A Journal Issued
  by the Courant Institute of Mathematical Sciences}, 55(11):1461--1506, 2002.

\bibitem{griniasty2019curved}
I.~Griniasty, H.~Aharoni, and E.~Efrati.
\newblock Curved geometries from planar director fields: Solving the
  two-dimensional inverse problem.
\newblock {\em Physical review letters}, 123(12):127801, 2019.

\bibitem{he2018rigid}
Z.~He and S.~D. Guest.
\newblock On rigid origami ii: Quadrilateral creased papers.
\newblock {\em arXiv preprint arXiv:1804.06483}, 2018.

\bibitem{huffman1976curvature}
D.~A. Huffman.
\newblock Curvature and creases: A primer on paper.
\newblock {\em IEEE Transactions on computers}, (10):1010--1019, 1976.

\bibitem{hull1994mathematics}
T.~Hull.
\newblock On the mathematics of flat origamis.
\newblock {\em Congressus numerantium}, pages 215--224, 1994.

\bibitem{hull2002modelling}
T.~C. Hull et~al.
\newblock Modelling the folding of paper into three dimensions using affine
  transformations.
\newblock {\em Linear Algebra and its applications}, 348(1-3):273--282, 2002.

\bibitem{izmestiev2017classification}
I.~Izmestiev.
\newblock Classification of flexible kokotsakis polyhedra with quadrangular
  base.
\newblock {\em International Mathematics Research Notices}, 2017(3):715--808,
  2017.

\bibitem{kasahara1998origami}
K.~Kasahara and T.~Takahama.
\newblock {\em Origami for the Connoisseur}.
\newblock Japan Pubns, 1998.

\bibitem{kawasaki1991relation}
T.~Kawasaki.
\newblock On the relation between mountain-creases and valley-creases of a flat
  origami.
\newblock In {\em Proceedings of the First International Meeting of Origami
  Science and Technology, 1991}, 1991.

\bibitem{kim2018printing}
Y.~Kim, H.~Yuk, R.~Zhao, S.~A. Chester, and X.~Zhao.
\newblock Printing ferromagnetic domains for untethered fast-transforming soft
  materials.
\newblock {\em Nature}, 558(7709):274--279, 2018.

\bibitem{klein2007shaping}
Y.~Klein, E.~Efrati, and E.~Sharon.
\newblock Shaping of elastic sheets by prescription of non-euclidean metrics.
\newblock {\em Science}, 315(5815):1116--1120, 2007.

\bibitem{kowalski2018curvature}
B.~A. Kowalski, C.~Mostajeran, N.~P. Godman, M.~Warner, and T.~J. White.
\newblock Curvature by design and on demand in liquid crystal elastomers.
\newblock {\em Physical Review E}, 97(1):012504, 2018.

\bibitem{kuribayashi2006self}
K.~Kuribayashi, K.~Tsuchiya, Z.~You, D.~Tomus, M.~Umemoto, T.~Ito, and
  M.~Sasaki.
\newblock Self-deployable origami stent grafts as a biomedical application of
  ni-rich tini shape memory alloy foil.
\newblock {\em Materials Science and Engineering: A}, 419(1-2):131--137, 2006.

\bibitem{lang1996computational}
R.~J. Lang.
\newblock A computational algorithm for origami design.
\newblock In {\em Proceedings of the twelfth annual symposium on Computational
  geometry}, pages 98--105, 1996.

\bibitem{lang2011origami}
R.~J. Lang.
\newblock {\em Origami design secrets: mathematical methods for an ancient
  art}.
\newblock AK Peters/CRC Press, 2011.

\bibitem{lang2018rigidly}
R.~J. Lang and L.~Howell.
\newblock Rigidly foldable quadrilateral meshes from angle arrays.
\newblock {\em Journal of Mechanisms and Robotics}, 10(2), 2018.

\bibitem{mahadevan2005self}
L.~Mahadevan and S.~Rica.
\newblock Self-organized origami.
\newblock {\em Science}, 307(5716):1740--1740, 2005.

\bibitem{miura1985method}
K.~Miura.
\newblock Method of packaging and deployment of large membranes in space.
\newblock {\em Title The Institute of Space and Astronautical Science Report},
  618:1, 1985.

\bibitem{mostajeran2016encoding}
C.~Mostajeran, M.~Warner, T.~H. Ware, and T.~J. White.
\newblock Encoding gaussian curvature in glassy and elastomeric liquid crystal
  solids.
\newblock {\em Proceedings of the Royal Society A: Mathematical, Physical and
  Engineering Sciences}, 472(2189):20160112, 2016.

\bibitem{na2015programming}
J.-H. Na, A.~A. Evans, J.~Bae, M.~C. Chiappelli, C.~D. Santangelo, R.~J. Lang,
  T.~C. Hull, and R.~C. Hayward.
\newblock Programming reversibly self-folding origami with micropatterned
  photo-crosslinkable polymer trilayers.
\newblock {\em Advanced Materials}, 27(1):79--85, 2015.

\bibitem{nash1954c1}
J.~Nash.
\newblock C1 isometric imbeddings.
\newblock {\em Annals of mathematics}, pages 383--396, 1954.

\bibitem{pellegrino2014deployable}
S.~Pellegrino.
\newblock {\em Deployable structures}, volume 412.
\newblock Springer, 2014.

\bibitem{plucinsky2018patterning}
P.~Plucinsky, B.~A. Kowalski, T.~J. White, and K.~Bhattacharya.
\newblock Patterning nonisometric origami in nematic elastomer sheets.
\newblock {\em Soft matter}, 14(16):3127--3134, 2018.

\bibitem{plucinsky2016programming}
P.~Plucinsky, M.~Lemm, and K.~Bhattacharya.
\newblock Programming complex shapes in thin nematic elastomer and glass
  sheets.
\newblock {\em Physical Review E}, 94(1):010701, 2016.

\bibitem{schenk2011origami}
M.~Schenk and S.~D. Guest.
\newblock Origami folding: A structural engineering approach.
\newblock {\em Origami}, 5:291--304, 2011.

\bibitem{schenk2013geometry}
M.~Schenk and S.~D. Guest.
\newblock Geometry of miura-folded metamaterials.
\newblock {\em Proceedings of the National Academy of Sciences},
  110(9):3276--3281, 2013.

\bibitem{silverberg2014using}
J.~L. Silverberg, A.~A. Evans, L.~McLeod, R.~C. Hayward, T.~Hull, C.~D.
  Santangelo, and I.~Cohen.
\newblock Using origami design principles to fold reprogrammable mechanical
  metamaterials.
\newblock {\em science}, 345(6197):647--650, 2014.

\bibitem{song2013enhanced}
Y.~Song, X.~Chen, V.~Dabade, T.~W. Shield, and R.~D. James.
\newblock Enhanced reversibility and unusual microstructure of a
  phase-transforming material.
\newblock {\em Nature}, 502(7469):85--88, 2013.

\bibitem{soorOri}
A.~Soor, P.~Velvaluri, P.~Plucinsky, and R.~D. James.
\newblock Origami design principles from abelian groups.
\newblock (In preparation).

\bibitem{tachi2009generalization}
T.~Tachi.
\newblock Generalization of rigid-foldable quadrilateral-mesh origami.
\newblock {\em Journal of the International Association for Shell and Spatial
  Structures}, 50(3):173--179, 2009.

\bibitem{tachi2009origamizing}
T.~Tachi.
\newblock Origamizing polyhedral surfaces.
\newblock {\em IEEE transactions on visualization and computer graphics},
  16(2):298--311, 2009.

\bibitem{tachi2010geometric}
T.~Tachi.
\newblock Geometric considerations for the design of rigid origami structures.
\newblock In {\em Proceedings of the International Association for Shell and
  Spatial Structures (IASS) Symposium}, volume~12, pages 458--460. Elsevier
  Ltd, 2010.

\bibitem{tolley2014self}
M.~T. Tolley, S.~M. Felton, S.~Miyashita, D.~Aukes, D.~Rus, and R.~J. Wood.
\newblock Self-folding origami: shape memory composites activated by uniform
  heating.
\newblock {\em Smart Materials and Structures}, 23(9):094006, 2014.

\bibitem{van2017growth}
W.~M. van Rees, E.~Vouga, and L.~Mahadevan.
\newblock Growth patterns for shape-shifting elastic bilayers.
\newblock {\em Proceedings of the National Academy of Sciences},
  114(44):11597--11602, 2017.

\bibitem{waitukaitis2019non}
S.~Waitukaitis, P.~Dieleman, and M.~van Hecke.
\newblock Non-euclidean origami.
\newblock {\em arXiv preprint arXiv:1909.13674}, 2019.

\bibitem{wei2013geometric}
Z.~Y. Wei, Z.~V. Guo, L.~Dudte, H.~Y. Liang, and L.~Mahadevan.
\newblock Geometric mechanics of periodic pleated origami.
\newblock {\em Physical review letters}, 110(21):215501, 2013.

\end{thebibliography}

\end{document}